\PassOptionsToPackage{final}{graphicx}
\PassOptionsToPackage{final}{hyperref}
\documentclass[sigconf,screen,final]{acmart}

\usepackage{ifthen}
\newboolean{spellingmode}
\setboolean{spellingmode}{false}
\newboolean{preprintmode}
\setboolean{preprintmode}{true}

\usepackage[utf8]{inputenc}

\AtBeginDocument{%
  }

\usepackage{mathtools}
\usepackage{ifdraft}
\usepackage{xcolor}
\usepackage[nomarginclue,nomargin,footnote]{fixme}
\usepackage{tikz}
\usepackage{thmtools}
\usepackage{coqsrcref}
\usepackage{twshowkeys}
\usepackage{fontawesome}
\usetikzlibrary{shapes,cd,backgrounds}

\usepackage{enumitem}
\setlist[enumerate,1]{label=(\arabic*),font=\normalfont,align=left,leftmargin=0pt,labelindent=0pt,listparindent=\parindent,labelwidth=0pt,itemindent=!,topsep=2pt,parsep=0pt,itemsep=2pt,start=1}
\setlist[enumerate,2]{label=(\alph*),font=\normalfont,labelindent=*,leftmargin=*,start=1,left=5pt,ref=(\arabic{enumi}\alph*)}
\setlist[itemize]{labelindent=*,leftmargin=*}
\setlist[description]{labelindent=*,leftmargin=*,itemindent=-1 em}

\FXRegisterAuthor{sm}{asm}{SM}%
\FXRegisterAuthor{tw}{twm}{TW}%
\makeatletter
\ifoptionfinal{
\@fxsetkeys{log}{silent=true}
}{}
\makeatother
\newcommand{\oldsmnote}[2][]{}
\newcommand{\oldtwnote}[2][]{}

\ifthenelse{\boolean{spellingmode}}{%
	\AtEndPreamble{
	\hyphenpenalty=50000 %
	\pagestyle{empty}%
	\thispagestyle{empty}%
	\def\pagestyle#1{}%
	\def\thispagestyle#1{}%
	\def\labelmarginpar#1{}%
	}
}{}

\usepackage{newunicodechar}
\newunicodechar{⇒}{\ensuremath{\Rightarrow}}
\newunicodechar{ℰ}{\ensuremath{\mathscr{E}}}
\newunicodechar{₀}{\ensuremath{{}_{0}}}
\newunicodechar{₁}{\ensuremath{{}_{1}}}
\newunicodechar{₂}{\ensuremath{{}_{2}}}

\newcommand{\takeout}[1]{\empty}

\newcommand{\monoto}{\rightarrowtail}
\newcommand{\epito}{\twoheadrightarrow}

\newcommand{\xra}[1]{\xrightarrow{~#1~}}

\newcommand{\pow}{\ensuremath{\mathscr{P}}}

\newcommand{\eps}{\varepsilon}

\newcommand{\USX}{U_X}%
\newcommand{\UX}{U_X}%
\newcommand{\UFA}{U_{FA}}%

\newcommand{\N}{\ensuremath{\mathbb{N}}}
\newcommand{\C}{\ensuremath{\mathscr{C}}}
\newcommand{\D}{\ensuremath{\mathscr{D}}}
\newcommand{\E}{\ensuremath{\mathscr{E}}}
\newcommand{\obSet}{\ensuremath{\mathscr{S}}}
\newcommand{\Set}{\ensuremath{\mathsf{Set}}}

\newcommand{\Coalg}{\ensuremath{{F\text{-}\mathsf{Coalg}}}}
\newcommand{\Alg}{\ensuremath{{F\text{-}\mathsf{Alg}}}}
\newcommand{\id}{\ensuremath{\mathsf{id}}}
\newcommand{\gen}[1]{\ensuremath{\langle{#1}\rangle}}
\newcommand{\inj}{\ensuremath{\mathsf{inj}}}
\newcommand{\inr}{\ensuremath{\mathsf{inr}}}
\newcommand{\inl}{\ensuremath{\mathsf{inl}}}

\DeclareMathOperator{\colim}{\ensuremath{\mathsf{colim}}}
\newcommand{\Fil}{\ensuremath{\mathsf{Fil}}}
\newcommand{\fp}{\ensuremath{\mathsf{p}}}
\newcommand{\finrec}{%
  \smash{\ensuremath{
    \overset{
    \scriptstyle\mathsf{fin}
    }{
    \scriptstyle\mathsf{rec}
    }
  }%
  }%
  }

\newcommand{\itemref}[2]{\autoref{#1}.\ref{#2}}
\newcommand{\textqt}[1]{`#1'}
\newcommand{\leafbullet}{\text{\color{black!80}\ensuremath{\bullet}}}
\colorlet{innernodecolor}{green!40!black}
\newcommand{\set}[2][]{%
  \ifthenelse{\equal{#2}{}}{%
    \ensuremath{{#1\emptyset}}%
  }{%
    \ensuremath{{#1\{#2#1\}}}%
  }%
}

\tikzset{
  bintree/.style={
    level distance=5mm,
    sibling distance=5mm,
    every node/.append style={
      inner sep=0pt,
      outer sep=0pt,
      minimum size=0pt,
    },
    edge from parent/.style={
      draw=innernodecolor,
      -,
      child anchor=center, %
      shorten >= 5pt,
    },
    subtree/.style={
      isosceles triangle,
      anchor=north,
      outer sep=0pt,
      inner sep=0pt,
      draw=black!60,
      shape border rotate=90,
      isosceles triangle stretches=true,
      inner sep=0pt, %
      minimum width=10pt,
      minimum height=10pt,
    },
  }
}
\newcommand{\cherryfunctor}[2][\set{\leafbullet}]{%
  #1 + #2 \mathcolor{innernodecolor}{\times} #2
}

\tikzstyle{reset attributes}=[
  minimum width=1mm,
  minimum height=1mm,
  text width=,
  anchor=center,
  align=none,
  solid,
  fill opacity=1.0,
  rounded corners=0pt,
]

\newcommand{\descto}[3][]{\arrow[phantom]{#2}[#1]{\text{\footnotesize{}\begin{tabular}{c}#3\end{tabular}}}}
\tikzset{shiftarr/.style={
    rounded corners,%
    to path={--([#1]\tikztostart.center)
      -- ([#1]\tikztotarget.center) \tikztonodes
      -- (\tikztotarget)},
  }}

\tikzset{shiftarr/.style={
    rounded corners,%
    to path={--([#1]\tikztostart.center)
      -- ([#1]\tikztotarget.center) \tikztonodes
      -- (\tikztotarget)},
  }}

\newcommand{\nicehref}[2]{%
\begin{tikzpicture}[baseline=(txt.base)]
\node[text=blue!80!white!40!black,inner sep=0pt,text depth=1.1pt] (txt)
  {\href{#1}{{#2}{\hspace*{1pt}\raisebox{.0pt}{\color{blue!30!white}\tiny\faExternalLink}}}};
\begin{scope}[on background layer]
\draw[color=blue!25!white] (txt.south west) -- (txt.south east);
\end{scope}
\end{tikzpicture}%
}

\newcommand{\sourceRepoURL}{%
  https://git8.cs.fau.de/software/initial-algebras-unchained%
}

\newcommand{\onlineHtmlURL}{%
  https://arxiv.org/src/2405.09504/anc/%
}
\newcommand{\softwareHeritageURL}{https://archive.softwareheritage.org/browse/origin/directory/?origin_url=https://git8.cs.fau.de/software/initial-algebras-unchained}

\newsavebox{\logoagdabox}
\sbox{\logoagdabox}{%
  \raisebox{-2pt}{\includegraphics[height=1em]{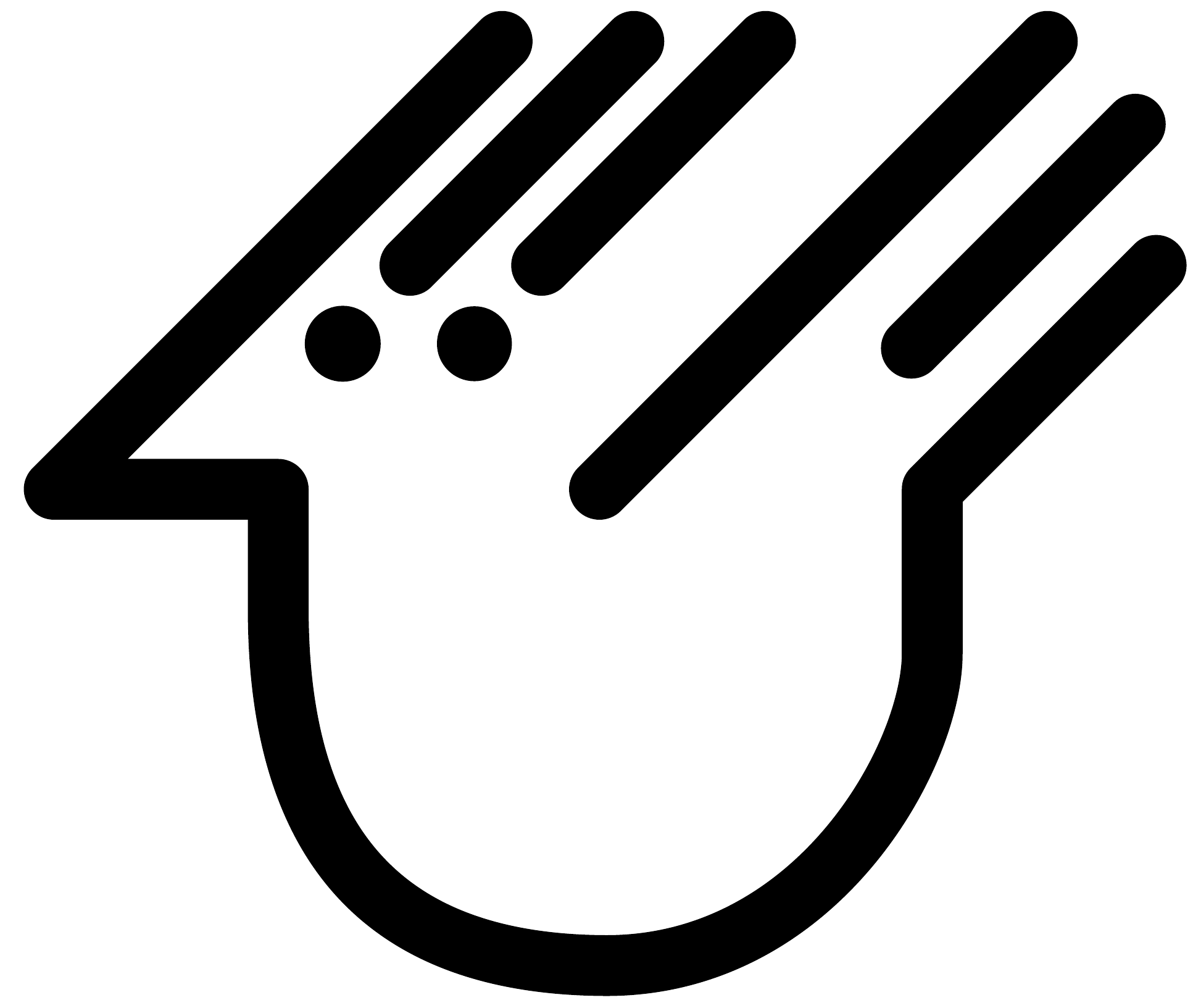}}%
}
\newcommand{\clickagdalogo}{%
  \href{\onlineHtmlURL index.html}{\usebox{\logoagdabox}}%
}
\newcommand{\agdaref}[4][]{%
  \coqref{#2}{#3}{#4%
  }{\href{\onlineHtmlURL #2.html\##3}{\usebox{\logoagdabox}}}}
\newcommand{\agdarefcustom}[4]{%
  \coqrefcustom{#1}{#2}{#3}{\href{\onlineHtmlURL #2.html\##3}{\usebox{\logoagdabox}}}%
  {#4%
  }}

\newcommand{\ownthmSpaceAbove}{5pt}
\newcommand{\ownthmSpaceBelow}{5pt}
\newcommand{\resetCurThmBraces}{%
  \gdef\curThmBraceOpen{(}%
  \gdef\curThmBraceClose{)}}
\resetCurThmBraces
\newcommand{\removeThmBraces}{%
  \gdef\curThmBraceOpen{}%
  \gdef\curThmBraceClose{}}

\makeatletter
\newenvironment{notheorembrackets}{%
  \removeThmBraces%
}{%
  \resetCurThmBraces%
}
\makeatother

\declaretheoremstyle[
spaceabove=\ownthmSpaceAbove,
spacebelow=\ownthmSpaceBelow,
headpunct=.,
postheadspace=.5em,
notebraces={\curThmBraceOpen}{\curThmBraceClose},
postheadhook={\resetCurThmBraces},
]{definition}
\declaretheoremstyle[
spaceabove=\ownthmSpaceAbove,
spacebelow=\ownthmSpaceBelow,
headpunct=.,
postheadspace=.5em,
notebraces={\curThmBraceOpen}{\curThmBraceClose},
postheadhook={\resetCurThmBraces},
bodyfont=\itshape,
]{mytheorem}
\makeatletter
\declaretheoremstyle[
style=definition,
preheadhook=\renewcommand\@upn{},
headpunct=.,
headfont=\it,
]{remark}
\makeatother

\theoremstyle{mytheorem}
\newtheorem{theorem}{Theorem}[section]
\newtheorem{lemma}[theorem]{Lemma}
\newtheorem{proposition}[theorem]{Proposition}

\newtheorem{corollary}[theorem]{Corollary}
\theoremstyle{definition}
\newtheorem{definition}[theorem]{Definition}
\newtheorem{example}[theorem]{Example}
\newtheorem{assumption}[theorem]{Assumption}
\newtheorem{notation}[theorem]{Notation}
\newtheorem{remark}[theorem]{Remark}

\AtEndPreamble{
}

\setcopyright{acmlicensed}
\copyrightyear{2024}
\acmYear{2024}
\acmDOI{}

\ifthenelse{\boolean{preprintmode}}{
\acmConference[Submission]{Submitted}{January 2024}{Easychair}
\acmISBN{}
}{
\setcopyright{acmlicensed}\acmConference[LICS '24]{39th Annual ACM/IEEE Symposium on Logic in Computer Science}{July 8--11, 2024}{Tallinn, Estonia}
\acmBooktitle{39th Annual ACM/IEEE Symposium on Logic in Computer Science (LICS '24), July 8--11, 2024, Tallinn, Estonia}
\acmDOI{10.1145/3661814.3662105}
\acmISBN{979-8-4007-0660-8/24/07}
}

\ifcsname mathcolor\endcsname%
\else%
\def\mathcolor#1#2{\color{#1}#2\color{black}}%
\fi%

\makeatletter
\ifthenelse{\boolean{preprintmode}}{
\settopmatter{printfolios=true,printacmref=false}
\def\@copyrightpermission{
\vspace{30mm} \\
Preprint}
}{}
\makeatother
\begin{document}
\title{Initial Algebras Unchained}
\subtitle{A Novel Initial Algebra Construction Formalized in Agda}

\author[T.~Wißmann]{Thorsten Wißmann}
\orcid{0000-0001-8993-6486}
\email{thorsten.wissmann@fau.de}
\affiliation{%
  \institution{Friedrich-Alexander-Universität Erlangen-Nürnberg}
  \streetaddress{Martensstr.\ 3}
  \city{Erlangen}
  \country{Germany}
  \postcode{91058}
}

\author[S.~Milius]{Stefan Milius}
\authornote{Funded by the Deutsche Forschungsgemeinschaft (DFG, German
Research Foundation) -- project numbers 470467389 and 517924115} %
\orcid{0000-0002-2021-1644}
\email{stefan.milius@fau.de}
\affiliation{%
  \institution{Friedrich-Alexander-Universität Erlangen-Nürnberg}
  \streetaddress{Martensstr.\ 3}
  \city{Erlangen}
  \country{Germany}
  \postcode{91058}
}

\begin{abstract}
  The initial algebra for an endofunctor $F$ provides a recursion and
  induction scheme for data structures whose constructors are
  described by $F$.  The initial-algebra construction by Ad\'amek
  (1974) starts with the initial object (e.g.~the empty set) and
  successively applies the functor until a fixed point is reached, an
  idea inspired by Kleene's fixed point theorem. Depending on the
  functor of interest, this may require transfinitely many steps
  indexed by ordinal numbers until termination.

  We provide a new initial algebra construction which is not based on
  an ordinal-indexed chain. Instead, our construction is loosely inspired by
  Pataraia's fixed point theorem and forms the colimit of all finite recursive
  coalgebras. This is reminiscent of the construction of the rational
  fixed point of an endofunctor that forms the colimit of \emph{all}
  finite coalgebras. For our main correctness theorem, we assume the
  given endofunctor is accessible on a (weak form of)
  locally presentable category.
  Our proofs are constructive and fully formalized in Agda.
\end{abstract}

\maketitle
\ifthenelse{\boolean{preprintmode}}{
\vspace{8mm}
}{}

\section{Introduction}
\twnote{}
Structural recursion is a fundamental principle in computer science; it appears whenever one traverses syntax or other inductively defined data structures such as natural numbers,
lists, or trees.
Any concrete recursion principle depends on the type of the constructors of the syntax or data structure of interest.
A uniform view on these recursion principles is provided by the theory of algebras for an endofunctor $F$ on a category, where $F$ is a parameter modelling the type of constructors.
An algebra for $F$ can be understood as an object of data with operations of the type modelled by $F$.

An initial algebra for $F$ -- if it exists -- provides the data structure which is the canonical minimal implementation of a data type with constructors specified by~$F$.
Its universal property yields a recursion and an induction principle for this data structure.

For example, for the set functor $F$ defined by $FX = \cherryfunctor{X}$ the initial $F$-algebra $(I,i)$ consists of all binary trees.
A standard example of a recursively defined function on binary trees is the \emph{height}-function $h\colon I \to \N$ defined by
\begin{equation}\label{eq:height}
  h(\bullet) = 0
  \qquad\text{and}\qquad
  h\big(
  \begin{tikzpicture}[reset attributes,bintree,baseline={([yshift=-2pt]s.north)}, level distance=4mm]
  \begin{scope}[edge from parent/.append style={child anchor=north, shorten >= 0pt}]
  \node {}
      child { node[subtree] (s) {\ensuremath{s}} }
      child { node[subtree] {\ensuremath{t}} } ;
  \end{scope}
  \end{tikzpicture}
  \big) = 1 + \max(h(s), h(t)),
\end{equation}
where the second case considers the binary tree obtained by joining the trees $s, t\in I$ under a new root node.
Initiality means that for every $F$-algebra there exists a unique homomorphism from $(I,i)$ to $(A,a)$.
For instance, the height-function $h$ is the unique function from the initial algebra $(I,i)$ to the $F$-algebra on $A= \N$ with the following structure $a$ (see also \autoref{fig:cherry-tree}): 
\begin{equation}
  a\colon F\N\to \N
  \quad
  a(\inl(\leafbullet)) = 0
  \quad
  a(\inr(k,n)) = 1+\max(k,n);
  \label{defheightalg}
\end{equation}
here $\inl\colon \set{\leafbullet} \to \cherryfunctor{\N}$ and $\inr\colon\N
\times N \to F\N$ denote the coproduct injections.
This models that a leaf has height $0$, and an inner node with children of height $k$ and $n$, respectively, has a height of $1+\max(k,n)$. 
\begin{figure}[b]
    \begin{tikzcd}[ampersand replacement=\&,row sep=10mm,column sep=3mm]
    \cherryfunctor{I}
    \arrow{d}{\cherryfunctor[\id_{\set{\leafbullet}}]{h}}
    \arrow{r}{i}
    \&[8mm] 
    I
    \arrow[dashed]{d}{\exists ! h}
    \& 
    \inr(\leafbullet,
  \begin{tikzpicture}[bintree,reset attributes,baseline=(current bounding box.south)]
  \node {}
      child { node {\leafbullet} }
      child { node {\leafbullet} } ;
  \end{tikzpicture}
    )
    \arrow[mapsto]{r}{i}
    \arrow{d}{\cherryfunctor[\id_{\set{\leafbullet}}]{h}}
    \&[8mm] 
  \begin{tikzpicture}[bintree,reset attributes,baseline={([yshift=-2pt]current bounding box.center)}]
  \node {}
    child { node {\leafbullet} }
    child { node {}
      child { node {\leafbullet} }
      child { node {\leafbullet} }
      } ;
  \end{tikzpicture}
  \arrow[mapsto]{d}{h}
    \\
    \cherryfunctor{\N}
    \arrow{r}{a}[below,pos=0.8]{\begin{array}{r@{\,}l}
      \leafbullet&\mapsto\, 0
      \\[-4pt]
      (k,n) &\mapsto\, 1+\max(k,n)
    \end{array}}
    \& \N
    \&
    \inr(0,1)
    \arrow[mapsto]{r}{a}
    \& 2
  \end{tikzcd}
  \caption{Example of an algebra $(\N,a)$ and the algebra homomorphism $h$ induced by the initial algebra $(I,i)$}
  \label{fig:cherry-tree}
\end{figure}
Note that the commutativity of the left-hand square in \autoref{fig:cherry-tree} is equivalent to the two equations in~\eqref{eq:height}; 
in particular, one sees that the arguments $k, n$ of the algebra structure~$a$ can be thought of as the returned values of the recursive calls of $h$ to the maximal subtrees $s$ and $t$ of a given binary tree which is not simply a leaf. 

Since initial algebras are an important concept, it is natural to investigate, when initial algebras exist and how they are constructed.
Their existence is entailed by assumptions on the `smallness' of the constructions performed by the endofunctor $F$. 
More precisely, every accessible endofunctor on a locally presentable category has an initial algebra; this follows from classical results by Ad\'amek~\cite{freealgebras}. 
He also provided a construction of the initial algebra for an endofunctor $F$ generalizing Kleene's
construction of the least fixed point of a continuous function on a cpo (or more generally, the corresponding transfinite version for monotone maps on chain-complete posets, which is essentially due to
Zermelo~\cite{Zermelo04}). 
For the category of sets, the construction builds the \emph{initial-algebra chain},%
\oldsmnote{@Thorsten: \emph{not} `initial chain'.} 
an ascending chain of sets, by transfinite recursion. 
It starts with the initial object, the empty set, and then successively applies the functor $F$:
\[
  \emptyset
  \to
  F\emptyset
  \to
  F^2\emptyset
  \to
  F^3\emptyset
  \to
  \cdots
\]
At a limit ordinal $i$, one takes the colimit of the chain of all previous sets. 
If the initial-algebra chain converges in the sense that a connecting map from one step to the next is an isomorphism, then its inverse is the structure of an initial algebra for $F$.

This provides a nice iterative construction of initial algebras.
However, the use of transfinite induction makes it difficult to
formalize this classical initial-algebra construction in full
generality in a proof assistant based on type theory, such as Agda,
Coq, or Lean. The problem seems to be that the notion of an ordinal is
inherently set-theoretic, and therefore does not easily translate into
type theory. It is the goal of our paper to provide a new construction
of the initial algebra which does not rely on transfinite recursion
and can be formalized in proof assistants.

Our construction is based on coalgebras obeying the recursion scheme~\cite{taylor1999}, aka.~\emph{recursive coalgebras}~\cite{osius1974,CaprettaUV06,Eppendahl99}. They are coalgebras $r\colon R \to FR$ satisfying a property much like the universal property of an initial algebra: for every $F$-algebra $b\colon FB\to B$, there exists be a unique coalgebra-to-algebra morphism:
\[
  \begin{tikzcd}
    FR
    \arrow{d}[swap]{Fh}
    &
    R \arrow{l}[swap]{r}
    \arrow[dashed]{d}{\exists! h}
    \\
    FB\arrow{r}{b}
    & B
  \end{tikzcd}
\]
Recursive coalgebras have originally arisen in the categorical study of well-founded induction~\cite{osius1974}; in fact, under mild assumptions on the endofunctor $F$ a coalgebra is recursive iff it is
\emph{well-founded}~\cite{Taylor96,Taylor23,taylor1999,JeanninKS17,AdamekMM20};%
\oldsmnote{@Thorsten: warum unklar über `strong connection' reden, wenn man es in einem Satz klar sagen kann?}
the latter notion generalizes the classical notion of a well-founded relation to the level of coalgebras for an endofunctor.

Consequently, a recursive coalgebra $R\to FR$ models a kind of decomposing or `divide' step within a recursive (divide-and-conquer) computation; this has been nicely explained by Capretta et al.~\cite{CaprettaUV06}.
Since it does not need be closed under operations of type~$F$, a
recursive coalgebra allows us to collect only some of the inhabitants
from the initial algebra or graph-like versions of the syntactic tree-like elements, e.g.~sharing subtrees.
For example,
\autoref{fig-ex-rec} shows an example of a recursive coalgebra for the set functor $FX=\cherryfunctor{X}$.
\begin{figure}
    \hfill
  \begin{tabular}{lcl@{}}
    \toprule
    $R$ & $r$ & $FR$ \\
    \midrule
    $u$ & $\mapsto$ & $\inr(x,x)$ \\
    $v$ & $\mapsto$ & $\inr(y,w)$ \\
    $w$ & $\mapsto$ & $\inr(z,y)$ \\
    $x$ & $\mapsto$ & $\inl(\leafbullet)$ \\
    $y$ & $\mapsto$ & $\inl(\leafbullet)$ \\
    $z$ & $\mapsto$ & $\inl(\leafbullet)$ \\
    \bottomrule
  \end{tabular}
  \hfill
  \begin{tikzpicture}[bintree,reset attributes,baseline=(current bounding box.center)]
  \node[label={left:$u~$}] (u) at (-15mm,0) {}
      child { node[label={left:$x~$}] (x) {\leafbullet} }
      child { node[label={right:$~x$}] {\leafbullet} } ;
  \node[label={left:$v~$}] (v) at (0,0) {}
      child { node[label={left:$y~$}] (x) {\leafbullet} }
      child { node[label={right:$~w$}] {}
        child { node[label={left:$z~$}] {\leafbullet} }
        child { node[label={right:$~y$}] {\leafbullet} }
      }
      ;
  \end{tikzpicture}
  \hfill
  \caption{Example of a recursive $F$-coalgebra $r\colon R\to FR$}
  \label{fig-ex-rec}
\end{figure}
This coalgebra is indeed recursive, because for every algebra $b\colon FB\to B$, the following function
\begin{align*}
  h\colon R&\longrightarrow B \\
  h(x) &:= h(y) := h(z) := b(\inl(\leafbullet)) \\
  h(u) &:= b(\inr(h(x),h(x))) \\
  h(w) &:= b(\inr(h(z),h(y))) \\
  h(v) &:= b(\inr(h(y),h(w)))
\end{align*}
specifies a unique coalgebra-to-algebra morphism from $(R,r)$ to $(B,b)$.
So recursive coalgebras can intuitively be understood as well-founded
data objects whose constructor types are specified by~$F$. 

An important special case is that if the structure $r\colon R \to FR$ of a recursive coalgebra happens to be an isomorphism, then $(R,r^{-1})$ is the initial algebra.

\subsection{Overview of the Contribution}
In the present work, we use this observation to provide a new
construction of the initial algebra. Intuitively, our construction is based on the idea that

\medskip
\begin{center}
  \parbox{7cm}{\itshape
  The initial algebra is the collection of well-founded data object of type $F$ (modulo behavioural equivalence).}
\end{center}

\medskip
Our construction works for every accessible endofunctor on a locally
presentable category. To make the point that our proof is constructive,
we formalize it in Agda. To this end we use the
notion of a $\Fil$-accessible category, for a collection $\Fil$ of filtered colimits,
which subsumes the notion of a locally
$\lambda$-presentable category (for a regular cardinal $\lambda$),
but without the need to mention any cardinal number $\lambda$ explicitly.

To construct an initial algebra, we consider recursive
coalgebras $r\colon R\to FR$ where $R$ is a \emph{presentable} object;
this notion generalizes that of a finite set. Therefore, such
coalgebras are said to be \emph{finite-recursive}.
We then take the colimit of all finite-recursive coalgebras for $F$ and
obtain a coalgebra $\alpha\colon A\to FA$. Even though $A$ is not
finitely presentable, $(A,\alpha)$ is
\emph{locally} finite-recursive, in the sense that it is built as a
colimit from finite-recursive coalgebras. We prove that $(A,\alpha)$
satisfies a universal property: it is the \emph{terminal}
locally finite-recursive coalgebra.

Finally, we prove that it is a fixed point of $F$. To this end, we provide a non-trivial
argument showing that $(FA,F\alpha)$ is also locally finite-recursive.
An argument in the style of Lambek's Lemma then shows that $\alpha$ is
an isomorphism, whence $(A,\alpha^{-1})$ is the initial algebra.

\paragraph{Structure of this paper.} After discussing categorical
preliminaries on locally presentable categories and recursive
coalgebras (\autoref{sec:prelim}), we establish our main initial
algebra theorem (\autoref{sec:main}). The relation and differences to
the construction based on the initial-algebra chain are discussed in
\autoref{sec:chain}. We address technical aspects and conceptual
challenges of the Agda formalization in \autoref{sec:agda}. The main
text concludes in \autoref{sec:conclusions}. The index of formalized
results (\autoref{agdarefsection}) links the results from the paper with
the corresponding result in Agda.

\subsection{Agda Formalization}
Our theorem is fully formalized in Agda (2.6.4) using the \texttt{ag\-da-cate\-go\-ries} library (v0.2.0)~\cite{agda-categories}.
Despite the effort, we went for the formalization for the following reasons:
\begin{enumerate}

\item Underpin our claim that our results are constructive. Agda makes
  it explicitly visible where non-constructive methods (such as the
  law of excluded middle) are used in our proofs or in
  proofs of standard lemmas (especially on hom-colimits and locally
  presentable categories).  In addition, the distinction between small
  and large colimits becomes visible due to Agda's type level system.

\item Exclude mistakes; especially the proof that $(FA,F\alpha)$ is 
  locally finite-recursive turned out to be non-trivial.

\item Contribute to the growing field of mechanized mathematics. We kept many
  lemmas on colimits and coalgebras as general as possible and
  will submit them to the \texttt{agda-categories} project.%

\end{enumerate}
The paper is phrased in standard set theory and category theory
to be more accessible to general readers, but we keep it as close to our Agda formalization as
possible.

Our Agda source code spans more than 5000 lines and 29 files. Both the source code and the HTML documentation can be found in the ancillary files on arxiv.org and on:
\begin{center}
\nicehref{\sourceRepoURL}{\sourceRepoURL}
\\
(also archived on \nicehref{\softwareHeritageURL}{archive.softwareheritage.org})
\end{center}

We annotate mechanized definitions and theorems with a
clickable icon \clickagdalogo{} which links
to the online HTML documentation of the respective result in the Agda code base.
So readers who want to have a quick look may browse the linked HTML
documentation.
Additional clues about the formalization (and clickable links) can be found in the index
in \autoref{agdarefsection} at the end of this document.

\subsection{Related Work}
Adámek et al.~\cite{AdamekMM21} have provided a proof of an initial algebra theorem which uses 
recursive coalgebras to construct the initial algebra from a given pre-fixed point of $F$ (that is, an algebra $A$ whose structure morphism is monomorphic).
They use Pataraia's fixed point theorem to obtain the initial algebra as the least fixed point of a monotone operator on the lattice of subobjects of $A$.
However, assuming existence of a pre-fixed point is quite strong; establishing this is almost as difficult as proving existence of an initial algebra.
Our construction works without that assumption and obtains (the carrier of) the initial algebra as a colimit.

Pitts and Steenkamp~\cite{PittsSteenkamp21} formalized an initial algebra theorem in Agda using
\emph{inflationary iteration} and avoiding transfinite iteration.

There have been efforts to incorporate ordinals in
Agda~\cite{ForsbergXG20} and transfinite induction in
Coq~\cite{transfinitecoq}. However, it is not clear to us whether
this could be used for any formalization of the initial-algebra chain
or the notion of a locally $\lambda$-presentable category.

Our construction is reminiscent of the construction of the rational
fixed point of a finitary endofunctor on a locally finitely
presentable category~\cite{AdamekMV06}. This is constructed by taking
the colimit of \emph{all} coalgebras with a finitely presentable
carrier (in lieu of all finite-recursive) coalgebras. In
addition, our technique of applying the functor to a locally finite-recursive coalgebra
(\autoref{iterate-locally-finrec}) is reminiscent of what is done in
work on the rational fixed
point~\cite{Milius10,mpw16,Urbat17,MiliusPW19} to show that
it is indeed a fixed point. Beside incorporating recursiveness, we
needed to slightly adapt
the ideas to make them work in the Agda formalization.

Our notion of a $\Fil$-accessible category (\autoref{D:Fil-acc}) is very similar to a notion used by Urbat~\cite{Urbat17}. It is also reminiscent of the notion of a $\mathbb D$-accessible category introduced by Ad\'amek et al.~\cite{AdamekEA02} and further studied by Centazzo et al.~\cite{Centazzo04, CentazzoEA04}.

\section{Categorical Preliminaries}\label{sec:prelim} 
We assume basic knowledge of category theory, functors, and co\-limits; see
\cite{awodey2010category,joyofcats} for a detailed introduction.
\begin{notation}
  Both in this paper and in the Agda source code, we use calligraphic letters
  $\C$, $\D$, $\E$ for categories, latin letters for functors and objects, and
  serif-free font for identifiers consisting of multiple letters, such as
  $\Set$ for the category of sets and maps.

  We denote the coproduct of objects $X,Y\in \C$ by $X+Y$ (if it
  exists) and we write $[f,g]\colon X+Y\to Z$ for the unique morphism
  induced by the morphisms $f\colon X\to Z$ and $g\colon Y\to Z$.
\end{notation}

\subsection{Finiteness in a category}
\label{sec:lfp}
A standard way to capture the notion of finiteness of an object in a category
is in terms of the preservation of filtered colimits:
\begin{definition}\label{lfpprelim} We recall from Ad\'amek and Rosick\'y~\cite{adamek1994locally}:
\begin{enumerate}%

\item\label{presColim} A functor $F\colon \C\to \C'$ \emph{preserves the colimit} $(c_i\colon
Di\to C)_{i\in \D}$ of a diagram $D\colon \D\to \C$ if $(Fc_i\colon FDi\to
FC)_{i\in \D}$ is a colimit of the diagram $F\circ D\colon \D\to \C'$
(\agdarefcustom{\itemref{lfpprelim}{presColim}}{Colimit-Lemmas}{preserves-colimit}{}).

\item\label{filtered} A category $\D$ is \emph{filtered} (\agdarefcustom{\itemref{lfpprelim}{filtered}}{Filtered}{filtered}{}) provided that
  \begin{enumerate}[left=8pt]
  \item $\D$ is non-empty,
  \item for every $X,Y\in \D$ there is an \emph{upper bound}
    $Z\in \D$, that is, there are morphisms $X\to Z$ and $Y\to Z$,
  \item for every $f,g \colon X\to Y$ there is some $Z \in \D$ and some $h\colon Y\to Z$ with $h\circ g = h\circ f$.
  \end{enumerate}
  A diagram $D\colon \D \to \C$ is \emph{filtered} if $\D$ is a filtered
  category, and a colimit of a filtered diagram is said to be \emph{filtered}.

\item A functor $F\colon \C\to \C'$ is \emph{finitary} if it preserves
  filtered colimits. 

\item An object $X\in \C$ is \emph{finitely presentable} if its hom-functor
$\C(X,-)\colon \C\to \Set$ is finitary.

\end{enumerate}
\end{definition}
\begin{remark}
  Note that the notions of small and large diagram schemes behave slightly
  differently in Agda than in classic set theory when quantifying over all
  sets. Therefore, we do not distinguish between small and large
  diagram schemes $\D$ above: preservation of colimits simply refers to those colimits
  that exist.
\end{remark}
\begin{remark}
  Finitariness can equivalently be defined using \emph{directed}
  diagrams (i.e.~those where the diagram scheme is a directed poset)
  in lieu of filtered ones~\cite[Thm.~1.5]{adamek1994locally}. We
  decided to work with filtered diagrams because this simplifies the
  formalization of locally finitely presentable categories (\autoref{D:lfp}).
\end{remark}
\begin{notheorembrackets}
\begin{example}[{\cite[Ex.\ 1.2]{adamek1994locally}}]
  Finite presentability instantiates to standard notions of finiteness in the following categories:
  \begin{enumerate}
  \item In the categories of sets, posets and graphs, the finitely
    presentable objects are precisely the finite sets, posets and
    graphs, respectively.
  \item In the category of vector spaces over a field $k$, the
    finitely presentable objects are precisely the finite dimensional
    vector spaces. 
  \item In the category of groups and monoids, the finitely
    presentable objects are precisely those which can be presented by
    finitely many generators and relations.
    
  \item More generally, in every finitary variety, i.e.~a category of algebras for a finitary
    signature satisfying a set of equations, the finitely presentable
    objects are precisely those presented by finitely many generators
    and relations. 
  \end{enumerate}
\end{example}
\end{notheorembrackets}
\begin{remark}\label{R:colim-set}
  Filtered colimits in $\Set$ can be characterized as follows. Given
  a filtered diagram $D\colon \D \to \Set$, a cocone $c_i\colon Di \to
  C$ ($i \in \D$) is a colimit of $D$ if and only if it satisfies the two conditions:
  \begin{enumerate}
  \item\label{cond:colim-set:1} the colimit injections are jointly surjective: for every $x
    \in C$ there exist $i \in \D$ and $x' \in Di$ such that $x =
    c_i(x')$;
  \item\label{cond:colim-set:2}
    for every pair $x', x'' \in
    Di$ with $c_i(x') = c_i(x'')$ there exists a morphism $h\colon i \to j$
    in $\D$ such that $Dh(x') = Dh(x'')$.
  \end{enumerate}

  General (non-filtered) colimits in $\Set$ have a similar characterization given by condition~\ref{cond:colim-set:1} and a more involved version of~\ref{cond:colim-set:2}.
  Condition \autoref{cond:colim-set:2} makes use of filteredness: whenever two elements are
  identified in the colimit, then are already identified by a
  connecting morphism of the diagram.

\end{remark}
When instantiating this characterization to filtered diagrams $D\colon \D\to
\C$ postcomposed with a hom-functor $\C(X,-)$ of some $X\in\C$, we obtain:
\begin{lemma}[\agdaref{Hom-Colimit-Choice}{hom-filtered-colimit-characterization}{}]\label{R:fpob}
  The hom-functor $\C(X,-)$ for an object $X$ preserves
  the colimit $C$ of a filtered diagram $D$ iff every morphism
  $f\colon X\to C$
  factorizes essentially uniquely through one of the colimit injection
  $c_i\colon Di \to C$ ($i \in \D$) in the sense that:
  \begin{enumerate}
  \item\label{R:fpob:1} there exist $i\in \D$ and $f'\colon X \to Di$ such that $f =
    c_i \circ f'$:
    \[
      \begin{tikzcd}[column sep=4mm]
        X
        \arrow{r}{\forall f}
        \arrow[dashed]{dr}[below left]{\exists i\text{ and }f'}
        &[5mm] C
        \descto{drr}{\(\Longleftrightarrow\)}
        &[15mm]&
        \mathllap{\forall~ }f
        \descto{r}{\normalsize\ensuremath{\displaystyle\in}}
        & \C(X,C)
        \\
        & Di
        \arrow{u}[swap]{c_i\text{ in }\C}
        &{}&
        \mathllap{\exists i\text{ with }~~}f'
        \arrow[mapsto]{u}[right]{\hspace*{1.3mm}\C(X,c_i)}
        \descto{r}{\normalsize\ensuremath{\displaystyle\in}}
        & \C(X,Di)
        \arrow{u}[right]{~~~\text{in }\Set}
      \end{tikzcd}
    \]
    
  \item\label{R:fpob:2} given two such factorizations $c_i \circ f' = c_i \circ f''$
    of $f$ there exists a morphism $h\colon i \to j$ of $\D$ such
    that $Dh \circ f' = Dh \circ f''$.
    \[
      \begin{tikzcd}[column sep=15mm,row sep=12mm]
        X
        \arrow{r}{c_i\circ f' = c_i\circ f''}
        \arrow[shift left=1]{dr}{f'}
        \arrow[shift right=1]{dr}[swap]{f''}
        & C
        \\
        & Di
        \arrow{u}[swap]{c_i}
        \arrow[dashed]{r}{Dh}
        & Dj
      \end{tikzcd}
    \]
  \end{enumerate}
\end{lemma}
\noindent
Indeed, apply \autoref{R:colim-set} to the cocone $\C(X,c_i)$ ($i \in \D$).

Next we recall the notion of a locally finitely presentable
category. The idea is that every object in such a category can be
constructed from a set of finitely presentable ones in a canonical way:
\begin{definition}[\agdaref{Canonical-Cocone}{}{The file defines the category, the forgetful functor, and the cocone \eqref{sliceCocone}}]\label{D:U}
For a set $\obSet$ of objects of $\C$ and $X\in \C$, define the category
$\obSet/X$ to have
\begin{itemize}[left=0pt]
\item objects $(S,f)$ for $S\in \obSet$ and $f\colon S\to X$ (in $\C$), and
\item morphisms $h\colon (S,f) \to (T,g)$ for $h\colon S\to T$ with $g\circ h
= f$ in $\C$.
\end{itemize}
The functor $\USX\colon \obSet/X\to \C$
is defined by $\USX(S,f) = S$.
Every such diagram $\USX$ has a canonical cocone:
\begin{equation}
  \USX(S,f) \overset{f}{\longrightarrow} X \qquad\text{for $(S,f) \in \obSet/X$}.
  \label{sliceCocone}
\end{equation}
\end{definition}
\begin{notheorembrackets}
\begin{definition}[{\cite{adamek1994locally}}]\label{D:lfp}
  A category $\C$ is \emph{locally finitely presentable} (\emph{lfp},
  for short) provided that it is cocomplete and has a set $\C_\fp$ of
  finitely presentable objects such that every object $X\in \C$ is
  the colimit of $\UX\colon \C_\fp/X \to \C$.
\end{definition}
\end{notheorembrackets}

\begin{example}
  Examples of lfp categories are ubiquitous. For example, the
  categories of sets, posets and graphs are lfp. Every finitary
  varieties of algebras forms an lfp category. Instances of this
  are the categories of groups, monoids and vector
  spaces and many others.
\end{example}

\subsection{Algebra and Coalgebra}

We recall some basic notions from the theory of (co-)algebras for an endofunctor. 
\begin{definition}
  Given an endofunctor $F\colon \C\to \C$, an $F$-algebra is a pair $(A,a)$ where $A$ is an object of $\C$ (the \emph{carrier} of the algebra) and with a morphism $a\colon FA\to A$ a morphism (its \emph{structure}).
  Dually, an $F$-coalgebra is a pair $(C,c)$ consisting of a carrier object $C$ and a structure morphism $c\colon C\to FC$.
\end{definition}

Intuitively, an $F$-algebra provides an object of data values
together with operations described by $F$. For example,
for the set
functor $FX = \cherryfunctor{X}$, an $F$-algebra $a\colon FA\to A$ consists of
\begin{itemize}[left=0pt]
\item some constant $a(\inl(\leafbullet)) \in A$, and
\item a binary operation sending $x,y\in A$ to $a(\inr(x,y)) \in A$.
\end{itemize}
An $F$-coalgebra $c\colon C\to FC$ provides an object of states with
abstract transitions or a structured collection of successors
described by $F$. For example, for the above set
functor $F$ we have for every state $x \in C$ that 
\begin{itemize}[left=0pt]
\item either $c(x) \in \set{\leafbullet}$, describing that $x$ has no successor,
\item or $c(x) \in C\mathcolor{innernodecolor}{\times} C$, describing
  that $x$ has a pair of successors.
\end{itemize}
So coalgebras can be thought of as graph like structures.

We relate $F$-algebras and $F$-coalgebras using the following notions of
morphisms:
\begin{definition}
  An \emph{$F$-algebra morphism} $h$ from $(A,a)$ to $(B,b)$ is a
  $\C$-morphism $h\colon A\to B$ satisfying $h\circ a = b\circ Fh$
  (see the right-hand square in diagram~\eqref{diag:mor}). An
  \emph{$F$-coalgebra morphism} $g$ from $(C,c)$ to $(D,d)$ is a
  $\C$-morphism $g\colon C\to D$ satisfying $d\circ g = Fg\circ c$
  (see the left-hand square in~\eqref{diag:mor}).  A
  \emph{coalgebra-to-algebra morphism} from a coalgebra $(D,d)$ to an
  algebra $(A,a)$ is a $\C$-morphism $s\colon D\to A$ such that
  $s = a\circ Fs\circ d$ (see the middle square in~\eqref{diag:mor}).
  \begin{equation}\label{diag:mor}
    \begin{tikzcd}
      C
      \arrow{r}{g}
      \arrow{d}[swap]{c}
      & D
      \arrow{r}{s}
      \arrow{d}[swap]{d}
      & A
      \arrow{r}{h}
      \arrow[<-]{d}{a}
      & B
      \arrow[<-]{d}{b}
      \\
      FC
      \arrow{r}{Fg}
      & FD
      \arrow{r}{Fs}
      & FA
      \arrow{r}{Fh}
      & FB
    \end{tikzcd}
  \end{equation}
\end{definition}

\begin{definition}
  A coalgebra $(R,r)$ is \emph{recursive} if for every algebra $(A,a)$ there is a unique coalgebra-to-algebra morphism from $(R,r)$ to $(A,a)$.
\end{definition}

In addition to the example of a recursive coalgebra for $FX=\cherryfunctor{X}$ in the introduction, we now discuss further examples which illustrate the point that recursive coalgebras relate to well-founded induction and that they capture the \textqt{divide} step in recursive divide-and-conquer computations.
\begin{example}
  \begin{enumerate}
  \item The first examples of recursive coalgebras are well-founded relations.
    Recall that a binary relation $R$ on a set $X$ is well-founded iff there is no infinite descending sequence of related elements:
    \[
      \cdots \mathbin{R} x_3 \mathbin{R} x_2 \mathbin{R} x_1 \mathbin{R} x_0.
    \]
    A binary relation $R\subseteq X \times X$ is, equivalently, a coalgebra for the power-set functor $\pow$ which maps a set $X$ to the set $\pow X$ of all its subsets: $R$ corresponds to $c\colon X \to \pow X$ defined by $c(x) = \set{y : y \mathbin{R} x}$.
    The relation $R$ is well-founded iff the associated $\pow$-coalgebra is recursive.

  \item For every endofunctor $F$ having an initial algebra $i\colon FI \to I$, its structure is an isomorphism (by Lambek's Lemma~\cite{lambek}), and the inverse $i^{-1}\colon I \to FI$ is a recursive coalgebra.

  \item Capretta et al.~\cite{CaprettaUV06} have shown how to obtain Quicksort using recursivity.
    Let $C$ be any linearly ordered set (of data elements). Quicksort is the recursive function $q\colon C^* \to C^*$ defined by
    \[
      q(\eps) = \eps \qquad\text{and}\qquad
      q(cw) = q(w_{\leq c}) \star (c q(w_{> c})),
    \]
    where $C^*$ is the set of all lists on $C$, $\eps$ is the empty list, $\star$ is the concatenation of lists and $w_{\leq c}$ denotes the list of those elements of~$w$ that are less than or equal to $c \in C$; analogously for $w_{> c}$.

    Now consider the set functor $FX =\set{\bullet} + C \times X \times X$ and define the coalgebra
    $s\colon C^* \to \set\bullet + C \times C^* \times C^*$ by
    \[
      s(\eps) = \bullet
      \quad\text{and}\quad
      s(cw) = (c, w_{\leq c}, w_{> c})
      \quad\text{for $c \in C$ and $w\in C^*$}.
    \]
    This coalgebra is recursive. Thus, for the following $F$-algebra $m\colon \set\bullet + C \times C^* \times C^* \to C^*$ defined by
    \[
      m(\bullet) = \varepsilon
      \qquad\text{and}\qquad
      m(c,w,v) = w \star (cv) 
    \]
    there exists a unique function $q$ on $C^*$ such that $q = m \circ Fq \circ s$.

    Note that this equation reflects the idea that Quicksort is a divide-and-conquer algorithm.
    The coalgebra structure $s$ divides a list into two parts $w_{\leq c}$ and $w_{> c}$.
    Then $Fq$ sorts these two smaller lists, and finally, in the combine (or conquer) step, the algebra structure $m$ merges the two sorted parts to obtain the desired whole sorted list.
  \end{enumerate}
\end{example}

Jeannin et al.~\cite[Sec.~4]{JeanninKS17} provide a number of further recursive functions arising in programming that are determined by recursivity of a coalgebra, e.g.~the Euclidean algorithm for the gcd of integers, Ackermann's function, and the Towers of Hanoi.

\begin{notation}
  We denote the category of $F$-coalgebras with their morphisms by $\Coalg$
  and the respective category of $F$-algebras by $\Alg$.
  The canonical forgetful functor $V\colon \Coalg\to\C$ is defined by $V(C,c) = C$ on objects, and it is identity on morphisms.
\end{notation}

\begin{lemma}[\agdaref{F-Coalgebra-Colimit}{}{We have various formulations,
depending on whether one needs the entire colimit or wants to prove that a
particular cocone in coalgebras is colimitting.}] \label{createCoalgColim}
The forgetful functor $V\colon \Coalg\to\C$ creates all colimits.
\end{lemma}
\noindent
This means that, given a diagram $D\colon \D\to\Coalg$, if the composed diagram $V\circ D\colon \D\to \C$ has a colimit $(h_i\colon VDi\to C)_{i\in \D}$, then there is a unique coalgebra structure $c\colon C \to FC$ such that all the colimit injections $h_i$ are coalgebra morphisms.
Moreover, $(C,c)$ is then a colimit of $D$. 
\begin{lemma}[\agdaref{Coalgebra.Recursive}{Limitting-Cocone-IsRecursive}{}] \label{recursiveColimit}
  Recursive coalgebras are closed under colimits created by $V$.
\end{lemma}
\begin{proof}[Proofsketch]
  Suppose that $D\colon \D\to \Coalg$ is a diagram such that $Dd = (C_d, c_d)$
  is recursive for every $d \in \D$, we have to show that its colimit
  $(C,c)$ is recursive, too. We denote the colimit injections by
  $i_d\colon (C_d, c_c) \to (C,c)$.

  Given an algebra $(A,a)$, we obtain for every $d \in D$ a unique
  coalgebra-to-algebra morphism $h_d$ from $(C,d)$ to $(A,a)$. It is
  not difficult to prove that $(h_d)$ is a cocone (of $VD$). Since the
  colimit of~$D$ is created by $V$, there exists a unique morphism
  $h\colon C \to A$ such that $h \circ i_d = h_d$ for all
  $d \in \D$. One now readily proves that $h$ is a unique
  coalgebra-to-algebra morphism from $(C,c)$ to $(A,a)$.
\end{proof}

A helpful lemma to prove recursiveness is the following
\begin{lemma}[\agdaref{Coalgebra.Recursive}{sandwich-recursive}{}, {\cite[Prop.\ 5]{CaprettaUV06}}]
  \label{sandwich-recursive}
  Consider coalgebra morphisms
  \[
  h\colon (R,r)\to (B,b) ~~\text{and}~~ g\colon (B,b)\to
  (FR,Fr)
  \]
  with $b = Fh\circ g$.
  If $(R,r)$ is recursive, then $(B,b)$ is recursive, too.
\end{lemma}

For $g=\id$, we immediately obtain a corollary stating that recursive coalgebras are closed
under application of $F$:
\begin{corollary}[{\agdaref{Coalgebra.Recursive}{iterate-recursive}{}, \cite[Prop.\ 6]{CaprettaUV06}}]
  \label{iterate-recursive}
  If $(R,r)$ is a recursive coalgebra, then so is $(FR,Fr)$.
\end{corollary}
\begin{lemma}[\agdaref{Coalgebra.Recursive}{iso-recursive⇒initial}{}, {\cite[Prop.\ 7(a)]{CaprettaUV06}}]
  \label{iso:recursive:initial}
  If the structure $r\colon R\to FR$ of a recursive coalgebra is an
  isomorphism, then $r^{-1}\colon FR\to R$ is the initial algebra.
\end{lemma}

\section{Initial Algebra Theorem}\label{sec:main}

Our main result (\autoref{cor:main}) states that for every accessible endofunctor $F$ on a locally presentable category $\C$ the initial algebra can be constructed as the colimit of all sufficiently `small' recursive coalgebras;
`small' here means that the carriers of those recursive coalgebra are $\lambda$-presentable, where $\lambda$ is a regular cardinal such that~$\C$ is locally $\lambda$-presentable and $F$ is $\lambda$-accessible.

Since in our Agda formalization the treatment of the cardinal number
$\lambda$ is problematic we work in a slightly generalized setting:
\begin{definition}[\agdaref{Accessible-Category}{Accessible}{}]\label{D:Fil-acc}
  Let $\Fil$ be a collection\footnote{In Agda, $\Fil$ is realized as a predicate; in classic set theory, it suffices to consider $\Fil$ as a class of small, filtered categories.} of filtered
  categories.

  \begin{enumerate}
  \item An object $X\in \C$ is (\Fil-)\emph{presentable} if its hom functor
    $\C(X,-)$ preserves colimits of diagrams
    $D\colon \D\to \C$ with $\D \in \Fil$.

  \item\label{D:Fil-acc:2} A category $\C$ is
    \emph{\Fil-accessible}
    provided that
    \begin{enumerate}
    \item there is a set $\C_\fp\subseteq \C$ of (\Fil-)presentable objects,
    \item for all $X\in \C$, the slice category $(\C_\fp/X)$ lies in $\Fil$,
    \item\label{i:canColimit} every object $X\in \C$ is the colimit
      of $\UX\colon \C_\fp/X \to \C$ (\autoref{D:U}).
    \end{enumerate}
  \end{enumerate}
\end{definition}
\begin{figure*}[t]
  \centering%
\begin{tikzpicture}[mynode/.style={
  align=center,
  draw=black,
  rounded corners=2pt,
},
mypath/.style={
  ->,
  line width=1pt,
  draw=black!80,
  rounded corners=2pt,
  shorten <= 2pt,
  shorten >= 2pt,
},
x=4.5cm,
y=3cm,
]
  \node[mynode] (diagram) at (0.2,0) {$\Coalg_{\finrec}$:\\
  Diagram of \\
  all recursive $(R,r)$ \\
  with $R\in \C_\fp$};
  \node[mynode] (Aa) at (0.9,1) {Locally finrec \\
  coalgebra $(A,\alpha)$};
  \path[mypath] (diagram) |- node[above] {Colimit} (Aa);
  \node[mynode] (inj) at (0.9,0) {Colimit Injection: \\
  it is the unique \\
  $h\colon (R,r) \to (A,\alpha)$ \\
  for every \\
  $(R,r)\in \Coalg_{\finrec}$};
  \path[mypath] (Aa) -- node[left] {\autoref{unique-proj}} (inj);

  \node[mynode] at (2,0) (ump) {Universal Property:\\
  there is a unique \\
  $(R,r) \overset{\exists!}{\longrightarrow} (A,\alpha)$ \\
  for all (locally) \\
  finrec $(R,r)$
  };
  \path[mypath] (inj) --
    node[above,yshift=1mm,] {\autoref{finrec-ump}}
    node[below,yshift=-1mm,] {\autoref{locally-finrec-ump}}
    (ump);

  \node[mynode] at (2,1) (iterate) {$(FA,F\alpha)$ is \\
  also a locally \\
  finrec coalgebra};
  \path[mypath] (Aa) -- node[above] {\autoref{iterate-locally-finrec}} (iterate);

  \node[mynode] (inverse) at (2.8,1) {There exists\\
    $h\colon (FA,F\alpha) \to (A,\alpha)$
  };

  \path[mypath] (iterate) -- (inverse);
  \path[mypath] (ump) -- (inverse.south west);
  \node[mynode] (uniq-endo) at (inverse |- ump) {The identity $\id_A$ \\
    is the only \\
    endomorphism \\
    $(A,\alpha) \to (A,\alpha)$
  };
  \path[mypath] (ump) -- (uniq-endo);

  \node[mynode] (iso) at (3.6,0.5) {Fixpoint\\
  $A\cong FA$};
  \path[mypath] (inverse) |-
        node[above,yshift=0mm,anchor=south west,pos=0.52,align=center]
            {Lambek's Lemma\\ (\autoref{lambek})}
            ([yshift=1mm]iso.west);
  \path[mypath] (uniq-endo) |- ([yshift=-2mm]iso.west);

  \node[mynode] (initial) at ([xshift=-2mm,yshift=-2mm]iso |- ump) {%
    Initial Algebra\\
    $(A,\alpha^{-1})$};
  \path[mypath] (iso.south) -- node[left] {\autoref{iso:recursive:initial}} (initial.north -| iso.south);
\end{tikzpicture}
  \caption{Roadmap for the initial algebra theorem (\autoref{main:thm})}
  \label{fig:roadmap}
\end{figure*}

\begin{remark}
  \begin{enumerate}
  \item \autoref{D:Fil-acc} is very similar to an instance of an $(\mathbb{I}, \mathcal M)$-accessible category in the sense of Urbat~\cite[Def.~3.4]{Urbat17}.
    His notion is parametric in a collection $\mathbb{I}$ of sifted diagrams and the class $\mathcal M$ of an $(\mathcal E, \mathcal M)$-factorization system which $\C$ is equipped with.
    Modulo the assumption on existence of colimits our notion is an instance of his by taking $\mathbb I = \Fil$ and the trivial factorization system $(\text{iso\-mor\-phisms},\text{all mor\-phisms})$. 

  \item \autoref{D:Fil-acc} is also reminiscent of the notion of a $\mathbb{D}$-accessible category introduced by Ad\'amek et al.~\cite[Def.~3.4]{AdamekEA02}.
    The parameter~$\mathbb{D}$ is a \emph{limit doctrine}, that is, an essentially small collection of small categories.
    They consider $\mathbb D$-filtered colimits, which are colimits that commute in $\Set$ with all limits with a diagram scheme in~$\mathbb D$.
    For example, for $\mathbb D$ consisting of all finite categories, one obtains filtered colimits, and for $\mathbb D$ consisting of all finite discrete categories, one obtains sifted colimits.
    The precise relationship of this notion to ours (or Urbat's notion) is subject to further study.

  \item Unlike the classical definition of a locally presentable category (and the generalized notions in the previous two items) our definition explicitly mentions the canonical diagram $\UX$, and we explicitly require that $\C_\fp/X$ lies in $\Fil$. The reason for this is that our proofs use that the hom functor of a presentable object preserves the colimit in \autoref{D:Fil-acc}.(2c).%
    \smnote[inline]{}
  \end{enumerate}
\end{remark}

\begin{example}\label{E:lfp}
  Every lfp category $\C$ is $\Fil$-accessible for $\Fil$ being the class of all filtered categories.
\end{example}
\begin{example}[\agdaref{Setoids-Accessible}{Setoids-Accessible}{}]
  For the same collection $\Fil$, the category $\Set$ is $\Fil$-accessible, with $\C_\fp$ being the category of all finite ordinals $\set{0,\ldots,k-1}$ for $k\in \N$ and all maps between them.
\end{example}
We have chosen to work with $\Fil$-accessible categories because their definition does not
explicitly mention cardinal numbers. However, besides lfp categories, they subsume the more general notion of a locally $\lambda$-presentable category~\cite{adamek1994locally}, for a regular cardinal $\lambda$:
\begin{example}\label{E:llp}
  Let $\lambda$ be a regular cardinal, and recall from op.~cit.~that a category $\D$ is \emph{$\lambda$-filtered} if every set of morphsism of size less than~$\lambda$ has a cocone in $\D$.
  Equivalently:
  \begin{enumerate}
  \item Every set of less than $\lambda$ objects has a cocone, and
  \item every family of less than $\lambda$ parallel morphisms $f_i\colon A\to B$ ($i\in I$ with $|I| < \lambda$) has a coequalizing morphism $g\colon B\to C$ (that is, $g\circ f_i\colon A\to C$ is independent of $i\in I$).
  \end{enumerate}
  A diagram $D\colon \D \to \C$ is \emph{$\lambda$-filtered} if so is its diagram scheme~$\D$, and a \emph{$\lambda$-filtered} colimit is a colimit of a $\lambda$-filtered diagram. A functor is \emph{$\lambda$-accessible} if it preserves $\lambda$-filtered colimits. An object $X$ of a category $\C$ is \emph{$\lambda$-presentable} if its hom-functor $\C(X,-)$ is $\lambda$-accessible. 

  Finally, a category is \emph{locally $\lambda$-presentable} if it is cocomplete and has a set of $\lambda$-presentable objects $\C_\fp$ such that every object of $\C$ is a $\lambda$-filtered colimit of the diagram $\UX\colon \C_\fp/X \to \C$ (cf.~\autoref{D:Fil-acc}).

  For $\Fil$ being the class of all small $\lambda$-filtered categories, every locally $\lambda$-presentable category is $\Fil$-accessible.
\end{example}

Note that like in lfp categories, the set $\C_\fp$ does not necessarily need
to contain all presentable objects (up to isomorphism):
\begin{lemma}[\agdaref{Accessible-Category}{presentable-split-in-fin}{}]\label{presentable-split}
  Every presentable object $X\in \C$ is a split quotient of some object $P\in \C_\fp$.
\end{lemma}
\noindent
That is, there exists a split epimorphism $P \epito X$.

\begin{assumption}\label{ass:main}
  For the remainder of this section, we fix a collection $\Fil$ of filtered categories, $\Fil$-accessible category $\C$ and an endofunctor $F\colon \C \to \C$.
  We also assume that every pair $X,Y$ of presentable objects has a coproduct $X+Y$. 
\end{assumption}
It then follows that $X+Y$ is presentable, too.
\begin{lemma}[\agdaref{Presentable}{presentable-coproduct}{}] \label{coproductPresentable}
  Presentable objects are closed under binary coproduct in $\C$.
\end{lemma}
\begin{proof}[Proofsketch]
  Similarly as in the setting of lfp categories, one uses filteredness (implied by
  $\Fil$) in order to lift the desired factorization property from presentable
  objects $X$ and $Y$ to $X+Y$.
\end{proof}

\begin{definition}\label{deffinrec}
  A coalgebra $(C,c)$ is 
  \begin{enumerate}%
  \item\label{i:finrec} \emph{finite-recursive} (\emph{finrec}) if it is recursive and $C$ presentable\twnote{} (\agdarefcustom{\itemref{deffinrec}{i:finrec}}{Iterate.Assumptions}{FiniteRecursive}{A record with a parameter (the coalgebra) and two members: 1.\ the carrier is presentable, 2.\ the coalgebra is recursive.}).
  \item \emph{locally finrec} if it is the colimit of finrec coalgebras (\agdarefcustom{\itemref{deffinrec}{i:localfinrec}}{CoalgColim}{}{\texttt{CoalgColim} describes a coalgebra that occurs as the colimit of coalgebras all satisfying a certain property. This property is then instantiated to \texttt{FiniteRecursive} in the main construction.}).
    This means that there is a diagram $D\colon \D\to \Coalg$ and a
    cocone $(Di\to (C,c))_{i\in \D}$ such that%
    \label{i:localfinrec}
    \begin{enumerate}
    \item $Di$ is finrec for every $i\in \D$,
    \item\label{i:carrier-colim} $(VDi\colon VDi\to C)_{i\in \D}$ is a
    colimit in $\C$ (where $V\colon \Coalg\to \C$ is the forgetful functor).
    \end{enumerate}
  \end{enumerate}
\end{definition}

Informally speaking, a finrec coalgebra is a system with transitions of type $F$ that is free of cycles and in which every element has a finite description.
This is precisely the intuition
of the elements of the initial algebra.
Consequently, our main theorem characterizes the initial
algebra (considered as a coalgebra) as the colimit $A$ of \emph{all} finrec
coalgebras. While we obtain the coalgebra structure $A\to FA$ directly from
colimit creation (\autoref{createCoalgColim}), we need some non-trivial results
about locally finrec coalgebras in order to construct its inverse viz.~the algebra structure
$FA\to A$. Note that as soon as we have established an
isomorphism $A\cong FA$, it is the initial algebra (by
\autoref{iso:recursive:initial}).

The way we obtain the desired isomorphism is very similar to the argument of
Lambek's famous lemma~\cite{lambek}. In dual form this argument can be stated as follows:
\begin{lemma}[\agdaref{Lambek}{lambek}{The proof is surprisingly short and readable.}]%
  \label{lambek}
  The structure $c\colon C\to FC$ of a coalgebra $(C,c)$ is an isomorphism provided that:
  \begin{enumerate}[topsep=0pt,left=0pt,midpenalty=99]
  \item\label{lambek:ex} there is a coalgebra morphism $h\colon (FC,Fc)\to (C,c)$,
  \item\label{lambek:uniq} there is at most one coalgebra morphism $(C,c)\to (C,c)$.
  \end{enumerate}
\end{lemma}
\noindent
The Agda formalization of the proof is quite easy to read;
in fact, this is a good place to start delving into our Agda code. 
\begin{proof}
  First note that the coalgebra structure is a coalgebra morphism
  $c\colon (C, c) \to (FC, Fc)$. Composing it with the coalgebra
  morphism $h\colon (FC, Fc) \to (C,c)$ according
  to~\autoref{lambek:ex} yields an endomorphism on $(C,c)$, which must
  be the identity by~\autoref{lambek:uniq}: $h \circ c = \id_C$. In
  the following computation we use this, functoriality of $F$, and
  that $h$ is a coalgebra morphism to conclude that it is the inverse
  of $c$:
  \[
    c \circ h = Fh \circ Fc = F(h \circ c) = F\id_C = \id_{FC}.
    \tag*{\qedhere}
  \]
\end{proof}
\begin{remark}
  Note that the (dual of the) original Lemma assumes that $(C,c)$ is
  terminal to establish the two conditions above.
\end{remark}

In order to verify the two condition of \autoref{lambek} for our
coalgebra $A \to FA$ we will develop elements of the theory of locally
finrec coalgebras. A roadmap of the major proof steps is visualized in
\autoref{fig:roadmap}.

\subsection{Applying $F$ to Locally Finrec Coalgebras}
\label{sec:iterate}%
At the heart of Lambek's lemma lies the observation that for every
coalgebra $c\colon C\to FC$ application of $F$ to the coalgebra
structure yields the coalgebra $Fc\colon FC\to FFC$.
This still holds for recursive coalgebras: If $(R,r)$ is
recursive, then so is $(FR,Fr)$ (\autoref{iterate-recursive}).
However, for finrec coalgebras this does not hold, not even for very
simple examples of set functors. For example, for the set functor
$FX=\N\times X$, the coalgebra $(FC, Fc)$ does not have a finite carrier unless
$C$ is empty. 

Generalizing to locally finrec coalgebras, we recover the ability to
apply $F$: if $(C,c)$ is locally finrec, then so is $(FC, Fc)$ provided
that $F$ preserves the colimit of a diagram of finrec coalgebra which
forms $(C,c)$. To stay in line with the Agda formalization, we
allow~$\D$ to be a large category.
\begin{assumption}[\agdaref{Iterate.ProofGlobals}{ProofGlobals}{We maintain a record of all running assumptions such that we can fill the namespace with one line (\texttt{open ProofGlobals}).}]
  \label{ass:iterate}
  We fix a locally finrec coalgebra $(A,\alpha)$ and a (potentially large)
  diagram $D\colon \D\to \Coalg$ of finrec coalgebras whose colimit is
  $(A,\alpha)$. We assume that $\D$ lies in $\Fil$, and that $F$
  preserves the colimit of the diagram $V\circ D\colon \D\to \C$. We write
  $(X_i,x_i)$ for the finrec coalgebra $Di$ and
  $\pi_i\colon (X_i,x_i) \to (A,\alpha)$ for the colimit injection
  (in $\Coalg$) for every $i\in \D$.
\end{assumption}
\begin{remark}
  The names $A$ and $\alpha\colon A\to FA$ are the members of the coalgebra
  record in the Agda's categories library. For our fixed locally finrec
  coalgebra we keep these names in order to stay
  as close as possible to the formalized proof, while all other
  coalgebras have roman letters as structures. 
\end{remark}

For an accessible functor $F$ on a locally presentable category $\C$
\renewcommand{\exampleautorefname}{Examples}
(\autoref{E:lfp} and~\ref{E:llp}),
\renewcommand{\exampleautorefname}{Example}
\autoref{ass:iterate} is satisfied: the diagram~$D$ may be taken as
the canonical projection
\[
  \Coalg_{\finrec}/(A,\alpha) \to \Coalg.
\]

In order to show that $(FA,F\alpha)$ is locally finrec again, we shall
present a diagram $E\colon \E\to \Coalg$ of finrec coalgebras
whose colimit is $(FA,F\alpha)$. For that we use that the object $FA$
is the colimit of the canonical diagram
\( \UFA \colon \C_\fp/FA\longrightarrow FA.  \)
(\itemref{D:Fil-acc}{i:canColimit}). The objects of $\C_\fp/FA$ are
pairs $(P,p)$ consisting of a presentable object $P\in \C_\fp$ and a
morphism $p\colon P\to FA$. The objects of the desired diagram scheme
$\E$
are commuting triangles with one side in $\C_\fp/FA$ and the
other side coming from a colimit injection of $D$, which we denote by $\E_0$:
\[
  \E_0 := \set{ (P,p,i,p') \mid
    \text{$(P,p) \in \C_\fp/FA$,
      $i\in \D$,
    $F\pi_i\circ p' = p$} }
  \tag{\agdarefcustom{Definition $\E_0$}{Iterate.FiniteSubcoalgebra}{ℰ₀}{}}
\]
\begin{equation}\label{eq:tri}
  (P,p,i,p') \in \E_0
  \qquad
  \hat{=}
  \qquad
  \begin{tikzcd}%
    |[alias=top]|
    P
    \arrow{r}{p}
    \arrow{dr}[swap]{p'}
    & FA
    \\
    & FX_i
    \arrow{u}[swap]{F\pi_i}
  \end{tikzcd}
\end{equation}
The elements of $\E_0$ relate the colimits of $D$ and
$\UFA$: the morphism $p$ is the colimit injection of $(P,p)$ to $FA$ and
$\pi_i$ is the colimit injection of $(X_i,x_i)$ to $(A,\alpha)$. The
collection $\E_0$ is only as large as $\D$: if $\D$ is a small
category, then $\E_0$ is a set.
\begin{lemma} \label{P-to-triangle}
  Every object of $\C_\fp/FA$ extends to an essentially unique object of $\E_0$:
  \begin{enumerate}[left=0pt]
  \item\label{i:tri-ex} For every $(P,p) \in \C_\fp/FA$ there exist $i\in \D$ and $p'\colon P\to FX_i$ 
    such that $(P,p,i,p') \in \E_0$. \hfill
    (\agdarefcustom{\itemref{P-to-triangle}{i:tri-ex}}{Iterate.FiniteSubcoalgebra}{P-to-triangle}{})
  \item\label{i:tri-uniq} For all $(P,p,i,p')\in \E_0$ and $p''\colon P\to FX_i$ with $F\pi_i\circ p'' = p$, there exist $j\in \D$ and a $\D$-morphism $i\to j$ making $p'$, $p''$ equal:
  \[
  \begin{tikzcd}[column sep=15mm]
    P
    \arrow[shift left=2]{r}{p'}
    \arrow[shift right=2]{r}[swap]{p''}
    & FX_i
    \arrow{r}{FD(i\to j)}
    & FX_j.
  \end{tikzcd}
  \tag*{(\agdarefcustom{\itemref{P-to-triangle}{i:tri-uniq}}{Iterate.FiniteSubcoalgebra}{CC.p'-unique}{})}
  \]
  \end{enumerate}
\end{lemma}
\begin{proof}
   Since $F$ preserves the colimit of $V \circ D\colon \D \to \C$ we
   know that $F\pi_i\colon FX_i \to FA$ ($i \in \D$) is a
   colimit. Now  apply \autoref{R:fpob} to this colimit and any given
   $(P,p)$ in $\C_\fp/FA$.
\end{proof}

Every object $t = (P,p,i,p') \in \E_0$ induces a finrec coalgebra:
since $P$ and $X_i$ are in $\C_\fp$, their coproduct $P+X_i$ exists
(\autoref{ass:main})
and is presentable (\autoref{coproductPresentable}). On $P+X_i$, we define the
following coalgebra structure
\[
  \gen{t}
  := \big(
  \begin{tikzcd}[column sep=14mm,cramped]
  P+X_i
  \arrow{r}{[p',x_i]}
  &[5pt] FX_i
  \arrow{r}{F\inr}
  & F(P+X_i)
  \end{tikzcd}
  \big)
\]
With \itemref{P-to-triangle}{i:tri-ex}, this coalgebra structure can be
understood to be generated by $p\colon P\to FA$. We put
\[
  E_0(t):= (P+X_i, \gen{t}) \in \Coalg.
\]
Moreover, we have a canonical coalgebra morphism
\[
  \inj_t := [p,\alpha\circ\pi_i]\colon  E_0(t) \longrightarrow (FA, F\alpha);
  \tag{\agdarefcustom{\ensuremath{\inj_t}}{Iterate.FiniteSubcoalgebra}{hom-to-FA}{Also see \texttt{hom-to-FA-i₁} and \texttt{hom-to-FA-i₂}}}
\]
indeed, consider the following diagram:
\[
  \begin{tikzcd}[column sep = 30]
    P + X_i
    \ar{r}{[p',x_i]}
    \ar{d}[swap]{[p,\alpha \circ \pi_i]}
    \ar[shiftarr = {yshift = 15}]{rr}{\gen{t}}
    &
    FX_i
    \ar{r}{F\inr}
    \ar{ld}[inner sep = 0]{F\pi_i}
    \ar{rd}[swap,inner sep = 0]{F(\alpha \circ \pi_i)}
    &
    F(P+X_i)
    \ar{d}{F[p,\alpha\circ \pi_i]}
    \\
    FA
    \ar{rr}{F\alpha}
    &&
    FFA
  \end{tikzcd}
\]
Its right-hand and middle parts commute due to functoriality of~$F$,
and for the left-hand part consider the coproduct components
separately: the left-hand one commutes by~\eqref{eq:tri} and the
right-hand one since $\pi_i$ is a coalgebra morphisms from $(X_i,x_i)$
to $(A,\alpha)$.

In order to turn $\E_0$ into the category $\E$ and $E_0$ into an appropriate diagram $\E\to
\Coalg$, the morphisms need
to be carefully chosen:
\begin{definition}[\agdaref{Iterate.DiagramScheme}{ℰ}{}]%
  \begingroup%
  \def\labelmarginpar#1{}%
  \label{D:E}%
  \endgroup%
  \begin{enumerate}
  \item The diagram scheme $\E$ is the full subcategory of
    $\Coalg/(FA,F\alpha)$ given by the coalgebra morphisms
    $\inj_t\colon E_0 t \to (FA,F\alpha)$ ($t\in \E_0$). In more
    detail:
    \begin{itemize}
    \item the objects of $\E$ are the elements of $\E_0$, and
    \item for $t_1,t_2\in \E_0$, the hom set $\E(t_1,t_2)$ consists of
      those coalgebra morphisms $h\colon E_0 (t_1)\to E_0 (t_2)$ such that
      $\inj_{t_2} \circ h = \inj_{t_1}$:
      \[
        \begin{tikzcd}[inj/.style={sloped,above,pos=0.35}]
          &[6mm]
          P+X_i
          \arrow{dd}{h}
          \arrow{r}{\gen{t_1}}
          \arrow{dl}[inj]{\inj_{t_1} =}[inj,below]{[p,\alpha\circ \pi_i]}
          & F(P+X_i)
          \arrow{dd}{Fh}
          \\
          FA
          \\
          &Q+X_j
          \arrow{r}{\gen{t_2}}
          \arrow{ul}[inj]{\inj_{t_2} =}[inj,below]{[q,\alpha\circ \pi_j]}
          & F(Q+X_j)
        \end{tikzcd}
        \text{for }
        \begin{array}{r@{}l}
        t_1 &=(P,p,i,p') \\
        t_2 &=(Q,q,j,q')
        \end{array}
      \]
    \end{itemize}

  \item The diagram $E\colon \E\to \Coalg$ is defined by%
    \[
      E(P,p,i,p') = (P+X_i, \gen{P,p,i,p'}), \qquad E(h) = h.
    \]
  \end{enumerate}
\end{definition}
Note that we do not need that $\E$ is filtered (or lies in $\Fil$); by
definition every colimit of finrec coalgebras is locally finrec. 
\begin{remark}
  \begin{enumerate}
  \item We obtain a canonical cocone on $E$ given by
    \[
      \inj_t\colon E(t)\to (FA,F\alpha)\qquad
      \text{for every $t\in \E_0$}.
    \]
  \item Note that it is important that $\E(t_1,t_2)$ contains
    \emph{all} coalgebra morphisms of type
    $h\colon P_1+X_{i_1}\to P_2+X_{i_2}$ and not just coproducts
    $h_\ell + h_r\colon P_1+X_{i_1} \to P_2+ X_{i_2}$. For in the
    latter case the colimit of $E$ would simply be the
    coproduct $FA + A$.
  \end{enumerate}
\end{remark}

All objects in the diagram $E$ are indeed recursive:
\begin{proposition}[\agdaref{Iterate.FiniteSubcoalgebra}{P+X-coalg-is-FiniteRecursive}{}]\label{allDiagRecursive}%
  For every object $t\in \E_0$, the coalgebra~$E(t)$ is recursive.
\end{proposition}
\begin{proof}
  For $t = (P,p,i,p')$, consider the following diagram:
  \[
    \begin{tikzcd}[column sep=12mm]
      X_i
      \arrow{r}{\inr}
      \arrow{dd}[swap]{x_i}
      \arrow{dr}{x_i}
      &
      P+X_i
      \arrow{r}{[p',x_i]}
      \arrow{d}{[p',x_i]}
      & FX_i
      \arrow{dd}{Fx_i}
      \arrow{dl}{\id}
      \\
      & FX_i
      \arrow{d}{F\inr}
      \arrow{dr}{Fx_i}
      \\
      FX_i
      \arrow{r}[swap]{F\inr}
      &
      F(P+X_i)
      \arrow{r}[swap]{F[p',x_i]}
      & FFX_i
    \end{tikzcd}
  \]
  It clearly commutes so that we obtain two coalgebra morphisms:
  \[
    (X_i,x_i) \xrightarrow{\inr} (P+X_i,\gen{t})
    \quad\text{and}\quad
    (P+X_i,\gen{t}) \xrightarrow{[p',x_i]} (FX_i,Fx_i)
  \]
  By definition, $\gen{t} = F\inr\circ [p',x_i]$,
  so all requirements of
  \autoref{sandwich-recursive} are met and thus $E(t) = (P+X_i,\gen{t})$ is recursive.
\end{proof}

For the verification that $(FA,F\alpha)$ is locally finrec, it remains to show
that the cocone $(\inj_t\colon VE t \to FA)_{t\in \E_0}$ is a colimit in~$\C$ (\itemref{deffinrec}{i:carrier-colim}).
We do so by relating it to the
colimits of $\C_\fp/FA$ and $D\colon \D\to \Coalg$. First a technical
lemma stating that $\D$-morphism can be extended to $\E$-morphisms:
\begin{lemma}[\agdaref{Iterate.DiagramScheme}{coalg-hom-to-ℰ-hom}{}]\label{L:mor}%
  For every $\D$-morphism $f\colon i \to j$ we have an $\E$-morphism
  \(
    \id_P + Df\colon (P,p, i, p') \to (P,p,j, FDf \circ p').
  \)
\end{lemma}
\begin{proof}
  Indeed, $h = Df\colon (X_i, x_i) \to (X_j, x_j)$ is a coalgebra
  morphism. Hence, the following diagram clearly commutes:
  \[
    \begin{tikzcd}[column sep = 40]
      P + X_i
      \ar{r}{[p',x_i]}
      \ar{d}[swap]{\id + h}
      &
      FX_i
      \ar{r}{F\inr}
      \ar{d}{Fh}
      &
      F(P+X_i)
      \ar{d}{F(\id + h)}
      \\
      P + X_j
      \ar{r}{[Fh \circ p', x_j]}
      &
      FX_j
      \ar{r}{F\inr}
      &
      F(P+X_j)
    \end{tikzcd}
  \]
  From $\pi_j \circ h = \pi_i$ we also obtain the commutative triangle
  \[
    \begin{tikzcd}[row sep = 5, baseline = (B.base)]
      P + X_i
      \ar{rd}{[p,\alpha \circ \pi_i]}
      \ar{dd}[swap]{\id + h}
      \\
      &
      FA
      \\
      |[alias = B]|
      P + X_j
      \ar{ru}[swap]{[p,\alpha\circ \pi_j]}
    \end{tikzcd}
    \tag*{\qedhere}
  \]
\end{proof}

\begin{lemma}[\agdaref{Iterate.Colimit}{cocone-is-triangle-independent}{}]
  \label{cocone-independent}
  For every $(V\circ E)$-cocone $(k_t\colon VE{t}\to K)_{t\in
    \E_0}$, the morphism $k_t$ only depends on $(P,p) \in
  \C_\fp/FA$. That is, for all objects $t_1 = (P,p,i_1,p_1)$ and
  $t_2 = (P,p,i_2,p_2)$  of $\E_0$ the following diagram commutes:
  \[
    \begin{tikzcd}[row sep=0mm, every cell/.append style={}]
      & P+X_{i_1}
      \arrow[bend left=15]{dr}[inner sep=0]{k_{t_1}}
      \\
      P
      \arrow[bend left=15]{ur}{\inl}
      \arrow[bend right=15]{dr}[swap]{\inl}
      && K
      \\
      & P+X_{i_2}
      \arrow[bend right=15]{ur}[swap]{k_{t_2}}
    \end{tikzcd}
  \]
\end{lemma}
\begin{proof}%
  Using that $\D$ is filtered (\autoref{ass:iterate}), we first take
  an upper bound $i_3$ of $i_1, i_2$ in $\D$, that is, we have
  $\D$-morphisms $h_1\colon i_1 \to i_3$ and $h_2\colon i_2 \to
  i_3$.
  Using them we can extend $i_3$ to an object of $\E_0$ in two ways:
  \[
    \begin{tikzcd}[column sep=14mm]
      P
      \arrow{r}{p}
      \arrow{d}[swap]{p_1}
      \arrow[rounded corners,to path={(\tikztostart.west) -- ++(-4mm,0)
        |- ([yshift=-4mm]\tikztotarget.south) \tikztonodes
        -- (\tikztotarget)}]{dr}[pos=0.25,left]{p_3 :=}
      & FA
      \\
      FX_{i_1}
      \arrow{ur}[description]{F\pi_{i_1}}
      \arrow{r}[swap]{FDh_1}
      & FX_{i_3}
      \arrow{u}[swap]{F\pi_{i_3}}
    \end{tikzcd}
    \begin{tikzcd}[column sep=14mm]
      P
      \arrow{r}{p}
      \arrow{d}[swap]{p_2}
      \arrow[rounded corners,to path={(\tikztostart.west) -- ++(-4mm,0)
        |- ([yshift=-4mm]\tikztotarget.south) \tikztonodes
        -- (\tikztotarget)}]{dr}[pos=0.25,left]{p_3' :=}
      & FA
      \\
      FX_{i_2}
      \arrow{ur}[description]{F\pi_{i_2}}
      \arrow{r}[swap]{FDh_2}
      & FX_{i_3}
      \arrow{u}[swap]{F\pi_{i_3}}
    \end{tikzcd}
  \]
  By the essential uniqueness (\itemref{P-to-triangle}{i:tri-uniq}),
  there exist an $i_4 \in \D$ and a $\D$-morphism $m\colon i_3\to i_4$ making $p_3$ and $p_3'$ equal:
  \[
    p_4 := \big(
    \begin{tikzcd}[column sep=15mm]
      P
      \arrow[shift left=2]{r}{p_3}
      \arrow[shift right=2]{r}[swap]{p_3'}
      & FX_{i_3}
      \arrow{r}{FDm}
      & FX_{i_4}
    \end{tikzcd}
    \big)
  \]
  Put $t_4 := (P,p,i_4,p_4)$ and use \autoref{L:mor} to extend the $\D$-morphisms $h_1$ and
  $h_2$ to the $\E$-morphisms $\id_P + h_i\colon t_i\to t_4$, $i = 1,2$.
  Using that the $k_t$ form a cocone, we can then verify the desired independence property:
  \[
    \begin{tikzcd}[column sep = 40,baseline=(bot.base)]
      & P+X_{i_1}
      \arrow[bend left=20]{dr}{k_{t_1}}
      \arrow{d}[swap]{\id+h_1}
      \\
      P
      \ar{r}{\inl}
      \arrow[bend left=20]{ur}{\inl}
      \arrow[bend right=20]{dr}[swap]{\inl}
      & P+X_{i_4}
      \arrow{r}{k_{t_4}}
      & K
      \\
      & |[alias=bot]|P+X_{i_2}
      \arrow{u}{\id+h_2}
      \arrow[bend right=20]{ur}[swap]{k_{t_2}}
    \end{tikzcd}
    \tag*{\qedhere}
  \]
\end{proof}

The independence following from \autoref{cocone-independent} allows us
to reduce cocones of $(V\circ E)$ to those of $\C_\fp/FA$. For the
latter we can then use the universal property of the colimit $FA$.
\begin{lemma}[\agdaref{Iterate.Colimit}{E-Cocone-to-D}{Here, $D$ refers to the
  canonical diagram $\UFA\colon \C_\fp/FA\to \C$}]\label{L:kbar}
  For every $(V\circ E)$-cocone $(k_t\colon VEt\to K)_{t\in \E_0}$, there is a
  $\UFA$-cocone $(\bar k_{(P,p)}\colon P\to K)_{(P,p)}$
  such that $\bar k_{(P,p)} = k_t \circ \inl$ for some $t= (P,p,i,p') \in \E_0$.
\end{lemma}
\begin{proof}%
  Given $(P,p)$, we define $\bar k_{(P,p)} := k_t\circ \inl$ for the
  $t\in \E_0$ obtained from \itemref{P-to-triangle}{i:tri-ex}. By
  \autoref{cocone-independent}, the morphism $\bar k_{(P,p)}$ is
  independent of the choice of $i$ and $p'$ in $t = (P,p,i,p')$.

  We now prove that every morphism $g\colon (P,p)\to (Q,q)$ in $\C_\fp/FA$ can be extended to an $\E$-morphism as follows: for every extension $t_2 = (Q, q, i, q') \in \E_0$ of $(Q,q)$ we have the morphism $g + \id_{X_i}$ in $\E$ from the extension $t_1 = (P,p, i, q' \circ g)$ of $(P,p)$.
  Indeed, $g + \id_{X_i}$ is a coalgebra morphism from $E(t_1)$ to $E(t_2)$:
  \[
    \begin{tikzcd}[column sep = 0mm]
      E(t_1)
      \equiv\mathrlap{\Big(}
      &
      P + X_i
      \ar{r}{[q' \circ g,x_i]}
      \ar{d}[swap]{g + \id}
      \ar[shiftarr = {yshift=15}]{rr}{\gen{t_1}}
      &[35pt]
      FX_i
      \ar{d}{\id}
      \ar{r}{F\inr}
      &[30pt]
      F(P + X_i)
      \ar{d}{F(g+\id)}
      \mathrlap{\smash{~~\Big)}}
      \\
      E(t_1)
      \equiv\mathrlap{\Big(}
      &
      Q + X_i
      \ar{r}{[q',x_i]}
      \ar[shiftarr = {yshift=-15}]{rr}{\gen{t_2}}
      &
      FX_i
      \ar{r}{F\inr}
      &
      F(Q+X_i)
      \mathrlap{\smash{~~\Big)}}
    \end{tikzcd}
  \]
  Moreover, using that $q \circ g = p$ we obtain $\inj_{t_2} \circ (g + \id_{X_i}) = \inj_{t_1}$:
  \[
    \begin{tikzcd}[row sep = 5]
      P + X_i
      \ar{rd}{[p,\alpha\circ \pi_i] = \inj_{t_1}}
      \ar{dd}[swap]{g + \id}
      \\
      &
      FA
      \\
      Q + X_i
      \ar{ru}[swap]{[q, \alpha \circ \pi_i] = \inj_{t_2}}
    \end{tikzcd}
  \]
  The cocone coherence property of $(k_t)_{t\in \E_0}$ then yields that of the morphisms $\bar k_{(P,p)}$; indeed, the following diagram commutes:
  \[
    \begin{tikzcd}[row sep = 5, baseline = (B.base)]
      P
      \ar{r}{\inl}
      \ar{dd}[swap]{g}
      \arrow[rounded corners,to path = {
        --([yshift=35]\tikztostart |- \tikztotarget.center)
        -- ([yshift=35]\tikztotarget.center)\tikztonodes 
        -- (\tikztotarget)
      }
      ]{rrd}{\bar k_{(P,p)}}
      &
      P + X_i
      \ar{rd}{k_{t_1}}
      \ar{dd}{g + \id}
      \\
      &&
      K
      \\
      Q
      \ar{r}{\inl}
      \arrow[rounded corners,to path = {
        --([yshift=-35]\tikztostart |- \tikztotarget.center)
        -- ([yshift=-35]\tikztotarget.center)\tikztonodes 
        -- (\tikztotarget)
      }
      ]{rru}{\bar k_{(Q,q)}}
      &
      |[alias = B]|
      Q+X_i
      \ar{ru}[swap]{k_{t_2}}
    \end{tikzcd}
    \tag*{\qedhere}
  \]
\end{proof}

For the verification of the universal property of the tentative colimit, we also
translate the cocone morphisms back and forth:
\begin{lemma}\label{lift-cocone-morph}
  For every $(V\circ E)$-cocone $(k_t\colon VEt\to K)_{t\in
    \E_0}$, a $\C$-morphism $v\colon FA\to K$ is a %
  morphism of $\UFA$-cocones, that is,
  \[
    \bar k_{(P,p)} = \big(
    P \xra{p} FA \xra{v} K
    \big),
    \qquad
    \text{for every $p  \in \C_\fp/FA$},
  \]
  if and only if it is a morphism of $(V \circ E)$-cocones:
  \[
    k_t = \big(
    P+X_i \xra{\inj_t} FA \xra{v}
    K
    \big),
    \qquad
    \text{for every $t \in \E$}.
  \]
\end{lemma}
\begin{proof}
  The implication for \textqt{if} is easy to verify (\agdarefcustom{\textqt{if}-direction of \autoref{lift-cocone-morph}}{Iterate.Colimit}{reflect-Cocone⇒}{}): we have
  \[
    v \circ p = v \circ \underbrace{[p,\alpha\circ \pi_i]}_{\inj_t} \circ \inl = k_t\circ \inl = \bar k_{(P,p)},
  \]
  using the definition of $\inj_t$ and $\bar k_{(P,p)}$ in the second and third steps, respectively.

  The argument for \textqt{only if} is non-trivial (\agdarefcustom{\textqt{only if}-direction of \autoref{lift-cocone-morph}}{Iterate.Colimit}{lift-Cocone⇒}{}).
  First recall from the proof of \autoref{allDiagRecursive} that $\inr\colon (X_i, x_i) \to (P+X_i, \gen{t})$ is a coalgebra morphism.
  Next we use that the object $\alpha \circ \pi_i\colon X_i \to FA$ of $\C_\fp/FA$ has the following factorization using that $\pi_i\colon (X_i,x_i) \to (A,\alpha)$ is a coalgebra morphism:
  \[
    \begin{tikzcd}
      X_i
      \ar{r}{\pi_i}
      \ar{rrd}[swap]{x_i}
      &
      A
      \ar{r}{\alpha}
      &
      FA
      \\
      &&
      FX_i
      \ar{u}[swap]{F\pi_i}
    \end{tikzcd}
  \]
  So we have the object $s = (X_i, \alpha\circ \pi_i, X_i, x_i)$ of $\E$ and we see that the codiagonal $\nabla\colon X_i + X_i \to X_i$ is a coalgebra morphism from $E(s)$ to $(X_i,x_i)$:
  \begin{equation}\label{diag:mor:2}
    \begin{tikzcd}[column sep = 30]
      X_i + X_i
      \ar{r}{[x_i,x_i]}
      \ar{d}[swap]{\nabla}
      \ar[shiftarr = {yshift=15}]{rr}{\gen{s}}
      &
      FX_i
      \ar{r}{F\inr}
      \ar{rd}{\id}
      &
      F(X_i+X_i)
      \ar{d}{F\nabla}
      \\
      X_i
      \ar{rr}{x_i}
      &&
      FX_i
    \end{tikzcd}
  \end{equation}
  Composing the coalgebra morphism $\inr\colon (X_i,x_i) \to (P+X_i, \gen{t})$ with the one in~\eqref{diag:mor:2} we obtain a morphism in $\E$ from $s$ to $t$;
  indeed, the composition is a coalgebra morphism $E(s) \to
  E(t)$, and we have $\inj_t \circ \inr \circ \nabla = \inj_s$:
  \[
    \begin{tikzcd}
      X_i + X_i
      \ar{d}[swap]{\nabla}
      \ar{rd}{[\alpha \circ \pi_i,\alpha \circ \pi_i]}
      \\
      X_i
      \ar{d}[swap]{\inr}
      \ar{r}{\alpha \circ \pi_i}
      &
      FA
      \\
      P + X_i
      \ar{ru}[swap]{[p,\alpha \circ \pi_i]}
    \end{tikzcd}
  \]
  So we have an $\E$-morphism $\inr\circ \nabla\colon s\to t$ and
  conclude that
  \begin{equation}\label{eq:ks}
    k_s = k_t \circ \inr \circ \nabla.
  \end{equation}
  We are ready to show the desired equation $k_t = v \circ \inj_t$. We
  consider the coproduct components of the domain $P + X_i$
  separately. For the left-hand component we have
  \begin{align*}
    v \circ \inj_t \circ \inl
    &=
    v \circ [p,\alpha \circ \pi_i] \circ \inl
    &
    \text{def.~of $\inj_t$}
    \\
    & =
    v \circ p
    &
    \text{since $[x,y]\circ \inl = x$}
    \\
    & = \bar k_{(P,p)}
    &
    \text{by assumption}
    \\
    & =
    k_t \circ \inl
    &
    \text{def.~of $\bar k_{(P,p)}$ (\autoref{L:kbar}).}
  \end{align*}
  For the right-hand coproduct component we compute as follows: 
  \begin{align*}
    v \circ \inj_t \circ \inr
    &=
    \mathrlap{v \circ [p,\alpha \circ \pi_i] \circ \inr
    = v \circ (\alpha \circ \pi_i)}
    &
    \\
    & =
    \bar k_{(X_i,x_i)}
    & \text{by assump.~for $p = \alpha\circ \pi_i$}
    \\
    & =
    k_s \circ \inl
    &
    \text{def.~of $\bar k_{(X_i,x_i)}$ (\autoref{L:kbar})}
    \\
    & =
    k_t \circ \inr \circ \nabla \circ \inl
    &
    \text{by~\eqref{eq:ks}}
    \\
    & =
    k_t \circ \inr
    &
    \text{since $\nabla \circ \inl = \id$}.
    \tag*{\qedhere}
  \end{align*}
\end{proof}

Finally, using that the canonical cocone
\[
  p\colon P \to FA
  \qquad ((P,p) \in \C_\fp/FA)
\]
is a colimit, we obtain that the $(V\circ E)$-cocone
\[
  \inj_t\colon VEt\to FA \qquad (t \in \E)
\]
is a colimit, too. Thus, under our \autoref{ass:iterate} we have the following result:
\begin{theorem}[%
  \agdaref{Iterate}{iterate-CoalgColimit}{The entire proof spreads over the following modules
  of \texttt{Iterate}:
  Assumptions, Colimit,  DiagramScheme,  FiniteSubcoalgebra,  ProofGlobals
}]\label{iterate-locally-finrec}
The coalgebra $(FA, F\alpha)$ is locally finrec. 
\end{theorem}
\begin{proof}
  We shall prove that the cocone $\inj_t\colon P+ X_i \to FA$ ($t \in \E$) is a colimit of $V \circ E$.
  Given a cocone $k_t\colon P+X_i \to K$ ($t \in \E$) we obtain the cocone $\bar k_{(P,p)}\colon P \to K$ ($(P,p) \in \C_\fp/FA$) using \autoref{L:kbar}.
  Thus, there exists a unique morphism $v\colon FA \to K$ such that $\bar k_{(P,p)} = v \circ p$ for every $p\colon P \to FA$ in $\C_\fp/FA$.
  By \autoref{lift-cocone-morph}, we have, equivalently, that $k_t = v \circ \inj_t$ for every $t \in \E$.
  Hence, $v$ is the unique morphism with this property, which completes the proof.
\end{proof}

\takeout{}%

\subsection{Unique Colimit Injections}
For the uniqueness condition in Lambek's lemma (\itemref{lambek}{lambek:uniq})
and the universal property in general, it helps to investigate when the
colimit injections of locally finrec coalgebras are unique as
coalgebra morphisms.

\takeout{}%

We continue to work under \autoref{ass:iterate}. 

\begin{lemma}[\agdaref{Unique-Proj}{hom-to-coalg-colim-triangle}{}]
  \label{coalg-hom-triangle}
  For every coalgebra $(B,\beta)$ with presentable carrier $B$,
  every coalgebra morphism $h\colon (B,\beta) \to (A,\alpha)$ factorizes through
  one of the colimit injections $\pi_j\colon (X_j,x_j)\to (A,\alpha)$ in $\Coalg$:
  \[
    \begin{tikzcd}
      (B,\beta)
      \arrow{r}{h}
      \arrow[dashed]{dr}[swap]{h'}
      & (A,\alpha)
      \\
      & (X_j,x_j)\arrow{u}[swap]{\pi_j}
    \end{tikzcd}
  \]
\end{lemma}
\begin{proof}%
  The hom-functor $\C(B,-)\colon \C\to \Set$ preserves the colimit $A$ of
  $V\circ D$; here, we use that we have a colimit in the base category
  $\C$ (\itemref{deffinrec}{i:carrier-colim}). Thus, we
  obtain an $i \in \D$ and a $\C$-morphism~$p'$ such that the
  following triangle commutes in $\C$ (cf.~\autoref{R:fpob}):
  \[
    \begin{tikzcd}
      B
      \arrow{r}{h}
      \arrow[dashed]{dr}[swap]{p'}
      & A
      \\
      &
      X_i%
      \arrow{u}[swap]{\pi_i}
    \end{tikzcd}
  \]
  Proving that $p'$ is a coalgebra morphism amounts to 
  showing that the left-hand square of the following diagram commutes:
  \[
    \begin{tikzcd}
      B
      \arrow{d}[swap]{\beta}
      \arrow{r}{p'}
      \descto{dr}{?}
      \arrow[shiftarr={yshift=5mm}]{rr}{h}
      & X_i
      \arrow{d}[swap]{x_i}
      \arrow{r}{\pi_i}
      \descto{dr}{\(\circlearrowleft\)}
      & A
      \arrow{d}{\alpha}
      \\
      FB
      \arrow{r}{Fp'}
      & FX_i
      \arrow{r}{F\pi_i}
      & FA
    \end{tikzcd}
  \]
  We will not prove its commutativity. Instead, observe that $F$
  preserves the colimit of $V\circ D$ (by assumption), so the
  morphisms $F\pi_i\colon FX_i \to FA$ ($i \in \D$) form a
  colimit. Since $\D$ lies in $\Fil$, it is filtered. We now use that
  $B$ is presentable and the ensuing essential uniqueness of
  factorizations of the morphism $\alpha\circ h\colon B\to FA$ through
  the colimit injection $F\pi_i$ (\itemref{R:fpob}{R:fpob:2}). The two
  paths of the left-hand square above are two such
  factorizations. Hence, there exists some morphism $d\colon i\to j$
  in $\D$ such that
  $FDd\circ x_i\circ p' = FDd \circ Fp'\circ \beta$:
  \[
    \begin{tikzcd}
      B
      \arrow{rr}{\alpha\circ h}
      \arrow[shift left=0]{dr}[sloped,above]{x_i\circ p'}
      \arrow[shift right=2]{dr}[sloped,below]{Fp'\circ \beta}
      & & FA
      \\
      & FX_i
      \arrow{ur}{F\pi_i}
      \arrow[dashed]{r}[below]{FDd}
      & FX_j
      \arrow{u}[swap]{F\pi_j}
    \end{tikzcd}
  \]
  We verify that $h' := Dd\circ p'\colon B\to X_j$ is the
  desired coalgebra morphism $(B,\beta)\to (X_j,x_j)$:
  \[
    \begin{tikzcd}
      B
      \ar{d}[swap]{\beta}
      \ar{r}{p'}
      \ar[shiftarr = {yshift=15}]{rr}{h'}
      &
      X_i
      \ar{r}{Dd}
      \ar{d}[swap]{x_i}
      &
      X_j
      \ar{d}{x_j}
      \\
      FB
      \ar{r}{Fp'}
      \ar[shiftarr = {yshift = -15}]{rr}{Fh'}
      &
      FX_i
      \ar{r}{FDd}
      &
      FX_j
    \end{tikzcd}
  \]
  Indeed, the right-hand square commutes, and the left-hand one does when postcomposed with $FDd$; thus the outside commutes as desired. Moreover, we have
  \[
    \pi_j \circ h' = \pi_j \circ Dd \circ p' = \pi_i \circ p' = h.
    \tag*{\qedhere}
  \]
\end{proof}
\begin{theorem}[\agdaref{Unique-Proj}{unique-proj}{The colimit injections are
called \emph{projections} in the Formalization because this is the terminology
in the \texttt{agda-categories} library}]\label{unique-proj}%
  Given a locally finrec coalgebra $(A,\alpha)$ and $i\in \D$, the colimit injection
  $\pi_i$ is the unique coalgebra morphism from $(X_i,x_i)$ to
  $(A,\alpha)$ provided that $D\colon \D\to \Coalg$ is full.
\end{theorem}
\begin{proof}%
  Given a coalgebra morphism $h\colon (X_i,x_i) \to (A,\alpha)$, apply
  \autoref{coalg-hom-triangle} to obtain a factorization through some colimit
  injection:
  \[
    \begin{tikzcd}
    & (A,\alpha)
    \\
    (X_i, x_i)
    \arrow{ur}{h}
    \arrow[dashed]{rr}{h'}
    &&
    (X_j,x_j) \arrow{ul}[swap]{\pi_j}
    \end{tikzcd}
  \]
  Since $D$ is full, $h' = Dd$ for some morphism $d\colon i\to j$ of~$\D$.
  Hence, $h = \pi_j\circ Dd = \pi_i$ by the cocone coherence condition.
\end{proof}

\subsection{Colimit of All Finrec Coalgebras}

We now move on to consider the colimit of all finrec coalgebras and establish that this satisfies the two conditions in \autoref{lambek}, which imply that it is the initial algebra.

\begin{notation}
  We denote by $\Coalg_{\finrec}$ the full subcategory of $\Coalg$ consisting of all finrec coalgebras with a carrier in $\C_\fp$.
\end{notation}
In the classical setting, this category is small:
The reason is that~$\C_\fp$ is a set, and on each object of $\C_\fp$ there is only a set of coalgebra structures.
So the colimit of the forgetful functor 
\[
\begin{tikzcd}
  \Coalg_{\finrec}
  \arrow[hook]{r}{D}
  & \Coalg
  \arrow[hook]{r}{V}
  & \C
\end{tikzcd}
\]
exists whenever $\C$ is cocomplete.
\begin{assumption}[\agdaref{Construction}{TerminalRecursive}{The assumptions are turned into module parameters. Instead of the essentially small $\Coalg_{\mathsf{rec}}^\mathsf{fin}$,
    the Agda code considers the colimit of those recursive coalgebras whose carrier lies in $\C_\fp$.}]
  For the remainder of this section we assume that $\Coalg_{\finrec}$ lies in $\Fil$, that the colimit of $V\circ D$ exists and is preserved by $F$.
\end{assumption}
\noindent
We denote $\colim(V \circ D)$ by $A$ and note that it carries a canonical coalgebra structure $\alpha\colon A \to FA$ such that $(A,\alpha) := \colim D$ (by \autoref{createCoalgColim}).
Moreover, the coalgebra $(A,\alpha)$ is locally finrec by definition.

From the uniqueness result in \autoref{unique-proj} we deduce the following universal property:
\begin{proposition}[\agdaref{Construction}{universal-property}{The statements about the constructed recursive coalgebra in the diagram are proven in retract-coalgebra-* in Coalgebra.Recursive}] \label{finrec-ump}
  For every finrec coalgebra $(C,c)$, there is a unique coalgebra morphism
  $(C,c)\to (A,\alpha)$.
\end{proposition}
\begin{proof}%
  By \autoref{presentable-split}, every presentable object $C\in \C$ is a split quotient of some object $P$ in $\C_\fp$ via $e\colon P \epito C$, say. Choose $m\colon C \monoto P$ such that $e\circ m = \id_C$.
  Then the following coalgebra structure 
  \[
    p= \big(
    P \xra{e}
    C \xra{c}
    FC \xra{Fm}
    FP\big)
  \]
  turns $e$ and $m$ into a coalgebra morphism; indeed, the diagram below commutes
  \[
    \begin{tikzcd}[column sep = 20]
      P
      \ar{r}{e}
      \ar{d}[swap]{e}
      &
      C
      \ar{r}{c}
      \ar{ld}{\id}
      &
      FC
      \ar{r}{Fm}
      \ar{rd}[swap]{\id}
      &
      FP
      \ar{d}{Fe}
      \\
      C
      \ar{d}[swap]{m}
      \ar{rrr}{c}
      \ar{rd}{\id}
      &&&
      FC
      \ar{d}{Fm}
      \ar{ld}[swap]{\id}
      \\
      P
      \ar{r}{e}
      &
      C
      \ar{r}{c}
      &
      FC
      \ar{r}{Fm}
      &
      FP
    \end{tikzcd}
  \]
  Thus, we have the coalgebra morphism $c \circ e\colon (P,p) \to (FC,Fc)$.
  By \autoref{sandwich-recursive} applied to $h = m$ and $g = c \circ e$ (and noting that $p = Fh \circ g$) we see that $(P,p)$ is recursive.
  Thus, this coalgebra lies in $\Coalg_{\finrec}$ since $P$ lies in
  $\C_\fp$.  The diagram
  $ D\colon \Coalg_{\finrec} \hookrightarrow \Coalg $ is a full
  functor. By \autoref{unique-proj}, the colimit injection
  $\pi\colon (P,p) \to (A,\alpha)$ is the unique coalgebra
  morphism. Using that $e \circ m = \id_C$, we see that there is a
  unique coalgebra morphism from $(C, c)$ to $(A,\alpha)$: we have the
  coalgebra morphism
  \[
    (C,c) \xra{m} (P,p) \xra{\pi} (A,\alpha),
  \]
  and given any coalgebra morphism $h\colon (C,c) \to (A,\alpha)$, we have $h \circ e = \pi$ by the unicity of $\pi$ whence $h = h \circ e \circ m = \pi \circ m$.
\end{proof}
\noindent
This universal property also lifts to colimits of finrec coalgebras:
\begin{corollary}[\agdaref{Construction}{universal-property-locally-finrec}{}]\label{locally-finrec-ump}%
  For every locally finrec coalgebra $(C,c)$, there is a unique coalgebra
  morphism $(C,c)\to (A,\alpha)$.
\end{corollary}
\begin{proof}
  \autoref{finrec-ump} lifts from finrec coalgebras to locally finrec coalgebras
  by a general property (\agdarefcustom{Colimit property in the proof of \autoref{locally-finrec-ump}}{Colimit-Lemmas}{colimit-unique-rep}{}) of colimits.
  If $D\colon \D\to \Coalg$ is the witnessing diagram $(C,c) = \colim D$ that consists of finrec coalgebras $Di$ ($i\in \D$),
  then there is a unique $Di\to (A,\alpha)$ for every $i\in \D$. Thus,
  $(A,\alpha)$ forms a cocone for $D$, which induces some morphism $(C,c)\to (A,\alpha)$.
  For uniqueness, consider $f,g\colon (C,c)\to (A,\alpha)$. For all $i\in
  \D$, we have $f\circ \inj_i = g\circ \inj_i\colon Di\to (A,\alpha)$, again by
  above uniqueness. Since colimit injections are jointly epic, this entails $f=g$.
\end{proof}
\begin{corollary} \label{thmTerminalLocallyFinrec}
  The coalgebra $(A,\alpha)$ is the terminal locally finrec coalgebra.
\end{corollary}

This allows us to prove our main theorem.
\begin{theorem}[\agdaref{Construction}{initial-algebra}{}]\label{main:thm}
  The coalgebra structure $\alpha\colon A\to FA$ is an isomorphism, and $\alpha^{-1}\colon FA\to A$ is the initial $F$-algebra.
\end{theorem}
\begin{proof}
  Applying $F$ to $(A,\alpha)$ yields a locally finrec coalgebra
  $(FA,F\alpha)$ (\autoref{iterate-locally-finrec}).
  By \autoref{locally-finrec-ump}, we obtain  a (unique) coalgebra morphism $(FA, F\alpha) \to (A,\alpha)$.
  By another application of \autoref{locally-finrec-ump}, we see that identity is the only coalgebra morphism on $(A,\alpha)$. 

  Thus, $\alpha$ is an isomorphism by Lambek's lemma (\autoref{lambek}), and by \autoref{iso:recursive:initial}, its inverse is the structure of the initial $F$-algebra.
\end{proof}
\begin{theorem}\label{cor:main}
  For every accessible endofunctor on a locally presentable category,
  the initial algebra is the colimit of all recursive coalgebras with a
  $\lambda$-presentable carrier.
\end{theorem}
\begin{proof}
  Suppose that $\C$ is locally $\lambda$-presentable and that
  $F\colon \C\to \C$ is $\lambda$-accessible. Let $\Fil$ be the class
  of all $\lambda$-filtered categories. Then $\Coalg_{\finrec}$ is an essentially small
  category consisting of all recursive coalgebra with a
  $\lambda$-presentable carrier and the colimit of
  \[
    \Coalg_{\finrec} \hookrightarrow \Coalg\overset{V}{\longrightarrow} \C
  \]
  exists and is preserved by $F$. Thus, $F$ has an initial algebra given by the colimit of the above diagram.
\end{proof}

\section{Comparison with the Initial-Algebra Chain}
\label{sec:chain}
The initial-algebra chain~\cite{freealgebras}, which we have recalled
for sets in the introduction, generalizes Kleene's fixed point
theorem: recall that the latter starts with the bottom element and then successively
applies a function to it. This yields an ascending chain approaching the
desired fixed point from below.

For the construction of the initial algebra for an endofunctor $F\colon \C\to \C$,
one starts with the initial object $0\in \C$. Its initiality induces a morphism
$!\colon 0\longrightarrow F0$.
Applying the functor successively to this morphism yields the $\omega$-chain
\[
  0\xra{!}
  F0\xra{F!}
  F^2 0\xra{F^2!} \cdots \xra{F^{k-1}!} F^k 0 \xra{F^k !}
  F^{k+1} 0 \xra{F^{k+1}!} \cdots
\]
Let us write $W_i$ for $F^i0$ and $w_{i,j}\colon W_i \to W_j$ for the
connecting morphisms. We also denote the colimit of this
$\omega$-chain by $W_\omega$ (assuming that it exists in $\C$).
\begin{enumerate}
\item By the universal property of this colimit there is a canonical morphism $W_\omega \to FW_\omega$.
  If the colimit is preserved by $F$ (e.g.~because~$F$ is finitary), then this morphism has an inverse which can be shown to be the structure of an initial $F$-algebra.
  
\item If the colimit is not preserved by $F$, then the iteration continues; the next step is induced by the universal property of the colimit~$W_\omega$:
  \[
    W_\omega
    \longrightarrow
    F W_\omega
    \longrightarrow
    FF W_\omega
    \longrightarrow
    \cdots
  \]
  The chain can be continued in this vein by transfinite recursion.
  If the functor $F$ is $\lambda$-\emph{accessible}~\cite[Def.~2.16]{adamek1994locally} for some regular cardinal $\lambda$, then the transfinite chain terminates in $\lambda$ steps; this means that $w_{\lambda,\lambda+1}\colon W_\lambda \to FW_\lambda$ is an isomorphism.
  In contrast, our construction does not use transfinite recursion and takes only one colimit, regardless of the size of $\lambda$.
\end{enumerate}

When looking at this chain through a coalgebraic lens, one observes that the chain consists of recursive coalgebras:
\begin{itemize}
\item $!\colon 0\to F0$ is trivially a recursive coalgebra by initiality.
\item Applying $F$ to this coalgebra yields recursive coalgebras $F^k !\colon F^k 0\to F(F^k 0)$, $k\in \N$ (\autoref{iterate-recursive}).
\item Their colimit $W_\omega \to FW_\omega$ is again recursive (\autoref{recursiveColimit}).
\end{itemize}

However, in general these recursive coalgebras are not contained in
the diagram scheme $\Coalg_{\finrec}$ that we use in our construction
in \autoref{sec:main}.  In every category, the initial object is
presentable, so $!\colon 0\to F0$ is a finrec coalgebra and thus a
split quotient of a finrec coalgebra in $\Coalg_{\finrec}$ by
\autoref{presentable-split}.  However, already the second step
$F0\to FF0$ of the initial-algebra chain may leave the realm of finrec
coalgebras, because $F0$ may not be presentable anymore.  Even for
simple set functors such as $FX = \N + X\times X$, the set
$F\emptyset = \N$ is infinite.

Of course, $F0\to FF0$ and, more generally, $F^k0 \to FF^k0$ for every $k \in \N$ are
\emph{locally} finrec; this follows from \autoref{iterate-locally-finrec}.

\section{Agda Formalization}\label{sec:agda}
The non-trivial details concerning the definition of the diagram $E\colon \E\to \Coalg$ (\autoref{D:E}) have motivated us to formalize the entire construction in a machine-checked setting.

\subsection{Technical Aspects}
After an attempt with Coq, we ultimately chose Agda (version 2.6.4) because it has an (almost official) library for category theory \cite{agda-categories} (version 0.2.0).
The formalized proofs are spread across 29 files and more than 5000 lines of code in total.
The entire source code and compilation instructions can be found
on
\begin{center}
\nicehref{\sourceRepoURL}{\sourceRepoURL}
\\
(also archived on \nicehref{\softwareHeritageURL}{archive.softwareheritage.org})
\end{center}
in the supplementary material archive.

All files compile with the flags
\texttt{\textendash{}\textendash{}without\textendash{}K} \texttt{\textendash{}\textendash{}safe}. For one file (\texttt{Iterate.Colimit}), we additionally use
\texttt{\textendash{}\textendash{}lossy\textendash{}unification} to adjust
Agda's unification heuristic and substantially speed up
compilation.\footnote{\url{https://agda.readthedocs.io/en/latest/language/lossy-unification.html}}
This does not compromise correctness.

\subsection{Formalization Challenges}
Agda's type system is organized in levels: if a structure (like a function or record) quantifies over all sets on level $\ell$, then the quantifying structure lives on level (at least) $\ell+1$.
This implies that if we consider coalgebras living on level $\ell$, then the property of being recursive lives on level $\ell+1$ because it quantifies over all algebras on level $\ell$.
Consequently, it is unclear whether a colimit of our main diagram $\Coalg_{\finrec}\hookrightarrow \Coalg$ exists, even when restricting to coalgebras over $\Set$.
However, assuming the law of excluded middle, we can bring $\Coalg_{\finrec}$ back down to level $\ell$ (\agdarefcustom{Level of $\Coalg_{\finrec}$}{Classical-Case}{IsRecursive-via-LEM}{}).
In order to keep the law of excluded middle out of the main construction, we allow potentially large colimits in $\Fil$-accessible categories and in the definition of locally finrec coalgebras.

Contrary to our original expectations, no issues regarding choice principles arose.
In our proofs, we have used multiple times that if a hom-functor preserves a colimit, then morphisms into the colimit factorize through the diagram (see e.g.~\autoref{sec:lfp} and \autoref{P-to-triangle}).
For the sake of modelling quotients, categories in the Agda library are not enriched over plain sets but instead over setoids.
The latter are sets with an explicit equivalence relation denoting element equality.
Thus, when working with elements of a colimit, in lieu of equivalence classes, one uses concrete representatives of equivalence classes.
Note that it was surprisingly tedious to prove that setoids forms a $\Fil$-accessible category (\agdarefcustom{Setoids are $\Fil$-accessible}{Setoids-Accessible}{Setoids-Accessible}{}).

\section{Conclusions and Future Work}\label{sec:conclusions}
We have shown that for a suitable endofunctor $F$ on a $\Fil$-accessible category, the initial algebra is obtained as the colimit of all recursive coalgebras with a presentable carrier. This formalizes the intuition that the initial algebra for $F$ is formed by all data objects of type $F$ modulo behavioural equivalence, which means that data objects are identified if they can be related by a coalgebra homomorphism.

Despite the fact that our description looks rather non-constructive, given that there is no concrete starting point, our main theorem can be proven in the constructive setting of Agda.

We leave as an open problem how well lfp categories can be formalized in a constructive setting with proper quotient types in lieu of setoids.
In addition, it would be interesting to see whether our construction can be adapted to well-founded coalgebras.

\begin{acks}
  We thank Sergey Goncharov for many helpful discussions on Agda and \texttt{agda-categories} throughout the entire formalization process.
  We also thank Lutz Schröder for interesting comments on the initial-algebra chain, Henning Urbat for inspiring discussions on fixed point theorems, and Nathanael Arkor for helpful comments on notions related to lfp categories.
  We have benefited a lot from helpful answers in the \texttt{\#agda} irc channel on \texttt{libera.chat}.
\end{acks}

\section{Index of formalized results}\label{agdarefsection}
Below we list the \textcolor{srcfilenamecolor}{Agda file} containing the referenced result and (if applicable) mention a concrete identifier in this file.
The respective HTML files can be found on
\begin{center}
\nicehref{\onlineHtmlURL index.html}{\onlineHtmlURL index.html}
\end{center}
and are also directly linked below.
\printcoqreferences

\bibliography{refs}


\begin{thebibliography}{27}


\ifx \showCODEN    \undefined \def \showCODEN     #1{\unskip}     \fi
\ifx \showDOI      \undefined \def \showDOI       #1{#1}\fi
\ifx \showISBNx    \undefined \def \showISBNx     #1{\unskip}     \fi
\ifx \showISBNxiii \undefined \def \showISBNxiii  #1{\unskip}     \fi
\ifx \showISSN     \undefined \def \showISSN      #1{\unskip}     \fi
\ifx \showLCCN     \undefined \def \showLCCN      #1{\unskip}     \fi
\ifx \shownote     \undefined \def \shownote      #1{#1}          \fi
\ifx \showarticletitle \undefined \def \showarticletitle #1{#1}   \fi
\ifx \showURL      \undefined \def \showURL       {\relax}        \fi
\providecommand\bibfield[2]{#2}
\providecommand\bibinfo[2]{#2}
\providecommand\natexlab[1]{#1}
\providecommand\showeprint[2][]{arXiv:#2}

\bibitem[Ad\'amek et~al\mbox{.}(2002)]%
        {AdamekEA02}
\bibfield{author}{\bibinfo{person}{Jir{\'{\i}} Ad\'amek},
  \bibinfo{person}{Francis Borceux}, \bibinfo{person}{Stephen Lack}, {and}
  \bibinfo{person}{Jir{\'{\i}} Rosick\'y}.} \bibinfo{year}{2002}\natexlab{}.
\newblock \showarticletitle{A classification of accessible categories}.
\newblock \bibinfo{journal}{\emph{J.~Pure Appl.~Algebra}}
  \bibinfo{volume}{175} (\bibinfo{year}{2002}), \bibinfo{pages}{7--30}.
\newblock


\bibitem[Ad{\'{a}}mek et~al\mbox{.}(2020)]%
        {AdamekMM20}
\bibfield{author}{\bibinfo{person}{Jir{\'{\i}} Ad{\'{a}}mek},
  \bibinfo{person}{Stefan Milius}, {and} \bibinfo{person}{Lawrence~S. Moss}.}
  \bibinfo{year}{2020}\natexlab{}.
\newblock \showarticletitle{On Well-Founded and Recursive Coalgebras}. In
  \bibinfo{booktitle}{\emph{Foundations of Software Science and Computation
  Structures - 23rd International Conference, {FOSSACS} 2020, Held as Part of
  the European Joint Conferences on Theory and Practice of Software, {ETAPS}
  2020, Dublin, Ireland, April 25-30, 2020, Proceedings}}
  \emph{(\bibinfo{series}{Lecture Notes in Computer Science},
  Vol.~\bibinfo{volume}{12077})}, \bibfield{editor}{\bibinfo{person}{Jean
  Goubault{-}Larrecq} {and} \bibinfo{person}{Barbara K{\"{o}}nig}} (Eds.).
  \bibinfo{publisher}{Springer}, \bibinfo{pages}{17--36}.
\newblock
\urldef\tempurl%
\url{https://doi.org/10.1007/978-3-030-45231-5\_2}
\showDOI{\tempurl}


\bibitem[Ad{\'{a}}mek et~al\mbox{.}(2021)]%
        {AdamekMM21}
\bibfield{author}{\bibinfo{person}{Jir{\'{\i}} Ad{\'{a}}mek},
  \bibinfo{person}{Stefan Milius}, {and} \bibinfo{person}{Lawrence~S. Moss}.}
  \bibinfo{year}{2021}\natexlab{}.
\newblock \showarticletitle{Initial Algebras Without Iteration ((Co)algebraic
  pearls)}. In \bibinfo{booktitle}{\emph{9th Conference on Algebra and
  Coalgebra in Computer Science, {CALCO} 2021, August 31 to September 3, 2021,
  Salzburg, Austria}}. \bibinfo{pages}{5:1--5:20}.
\newblock
\urldef\tempurl%
\url{https://doi.org/10.4230/LIPICS.CALCO.2021.5}
\showDOI{\tempurl}


\bibitem[Ad{\'{a}}mek et~al\mbox{.}(2006)]%
        {AdamekMV06}
\bibfield{author}{\bibinfo{person}{Jir{\'{\i}} Ad{\'{a}}mek},
  \bibinfo{person}{Stefan Milius}, {and} \bibinfo{person}{Jiri Velebil}.}
  \bibinfo{year}{2006}\natexlab{}.
\newblock \showarticletitle{Iterative algebras at work}.
\newblock \bibinfo{journal}{\emph{Math. Struct. Comput. Sci.}}
  \bibinfo{volume}{16}, \bibinfo{number}{6} (\bibinfo{year}{2006}),
  \bibinfo{pages}{1085--1131}.
\newblock
\urldef\tempurl%
\url{https://doi.org/10.1017/S0960129506005706}
\showDOI{\tempurl}


\bibitem[Adámek(1974)]%
        {freealgebras}
\bibfield{author}{\bibinfo{person}{Jiří Adámek}.}
  \bibinfo{year}{1974}\natexlab{}.
\newblock \showarticletitle{Free algebras and automata realizations in the
  language of categories}.
\newblock \bibinfo{journal}{\emph{Commentationes Mathematicae Universitatis
  Carolinae}} \bibinfo{volume}{015}, \bibinfo{number}{4}
  (\bibinfo{year}{1974}), \bibinfo{pages}{589--602}.
\newblock
\urldef\tempurl%
\url{http://eudml.org/doc/16649}
\showURL{%
\tempurl}


\bibitem[Adámek et~al\mbox{.}(2004)]%
        {joyofcats}
\bibfield{author}{\bibinfo{person}{Jiří Adámek}, \bibinfo{person}{Horst
  Herrlich}, {and} \bibinfo{person}{George~E. Strecker}.}
  \bibinfo{year}{2004}\natexlab{}.
\newblock \bibinfo{title}{Abstract and Concrete Categories. The Joy of Cats}.
\newblock
\newblock


\bibitem[Adámek and Rosick{\'y}(1994)]%
        {adamek1994locally}
\bibfield{author}{\bibinfo{person}{Jiří Adámek} {and}
  \bibinfo{person}{Jiří Rosick{\'y}}.} \bibinfo{year}{1994}\natexlab{}.
\newblock \bibinfo{booktitle}{\emph{Locally Presentable and Accessible
  Categories}}.
\newblock \bibinfo{publisher}{Cambridge University Press}.
\newblock


\bibitem[Awodey(2010)]%
        {awodey2010category}
\bibfield{author}{\bibinfo{person}{Steve Awodey}.}
  \bibinfo{year}{2010}\natexlab{}.
\newblock \bibinfo{booktitle}{\emph{Category Theory}}.
\newblock \bibinfo{publisher}{OUP Oxford}.
\newblock
\showISBNx{9780199587360}
\showLCCN{2010483708}
\urldef\tempurl%
\url{http://books.google.de/books?id=-MCJ6x2lC7oC}
\showURL{%
\tempurl}


\bibitem[Capretta et~al\mbox{.}(2006)]%
        {CaprettaUV06}
\bibfield{author}{\bibinfo{person}{Venanzio Capretta}, \bibinfo{person}{Tarmo
  Uustalu}, {and} \bibinfo{person}{Varmo Vene}.}
  \bibinfo{year}{2006}\natexlab{}.
\newblock \showarticletitle{Recursive coalgebras from comonads}.
\newblock \bibinfo{journal}{\emph{Inf. Comput.}} \bibinfo{volume}{204},
  \bibinfo{number}{4} (\bibinfo{year}{2006}), \bibinfo{pages}{437--468}.
\newblock
\urldef\tempurl%
\url{https://doi.org/10.1016/j.ic.2005.08.005}
\showDOI{\tempurl}


\bibitem[Centazzo(2004)]%
        {Centazzo04}
\bibfield{author}{\bibinfo{person}{Claudia Centazzo}.}
  \bibinfo{year}{2004}\natexlab{}.
\newblock \emph{\bibinfo{title}{Generalized algebraic models}}.
\newblock \bibinfo{thesistype}{Ph.\,D. Dissertation}.
  \bibinfo{school}{Département de Mathématique de la Faculté des Sciences de
  l’Université catholique de Louvain}.
\newblock


\bibitem[Centazzo et~al\mbox{.}(2004)]%
        {CentazzoEA04}
\bibfield{author}{\bibinfo{person}{Claudia Centazzo},
  \bibinfo{person}{Jir{\'{\i}} Rosick\'y}, {and} \bibinfo{person}{Enrico
  Vitale}.} \bibinfo{year}{2004}\natexlab{}.
\newblock \showarticletitle{A characterization of locally $D$-presentable
  categories}.
\newblock
  \bibinfo{journal}{\emph{Cah.~Topol.~G\'{e}om.~Diff\'{e}r.~Cat\'{e}g.}}
  \bibinfo{volume}{45}, \bibinfo{number}{2} (\bibinfo{year}{2004}),
  \bibinfo{pages}{141--146}.
\newblock


\bibitem[Eppendahl(1999)]%
        {Eppendahl99}
\bibfield{author}{\bibinfo{person}{Adam Eppendahl}.}
  \bibinfo{year}{1999}\natexlab{}.
\newblock \showarticletitle{Coalgebra-to-Algebra Morphisms}. In
  \bibinfo{booktitle}{\emph{Conference on Category Theory and Computer Science,
  {CTCS} 1999, Edinburgh, UK, December 10-12, 1999}}
  \emph{(\bibinfo{series}{Electronic Notes in Theoretical Computer Science},
  Vol.~\bibinfo{volume}{29})}, \bibfield{editor}{\bibinfo{person}{Martin
  Hofmann}, \bibinfo{person}{Giuseppe Rosolini}, {and} \bibinfo{person}{Dusko
  Pavlovic}} (Eds.). \bibinfo{publisher}{Elsevier}, \bibinfo{pages}{42--49}.
\newblock
\urldef\tempurl%
\url{https://doi.org/10.1016/S1571-0661(05)80305-6}
\showDOI{\tempurl}


\bibitem[Forsberg et~al\mbox{.}(2020)]%
        {ForsbergXG20}
\bibfield{author}{\bibinfo{person}{Fredrik~Nordvall Forsberg},
  \bibinfo{person}{Chuangjie Xu}, {and} \bibinfo{person}{Neil Ghani}.}
  \bibinfo{year}{2020}\natexlab{}.
\newblock \showarticletitle{Three equivalent ordinal notation systems in
  cubical Agda}. In \bibinfo{booktitle}{\emph{Proceedings of the 9th {ACM}
  {SIGPLAN} International Conference on Certified Programs and Proofs, {CPP}
  2020, New Orleans, LA, USA, January 20-21, 2020}}. \bibinfo{pages}{172--185}.
\newblock
\urldef\tempurl%
\url{https://doi.org/10.1145/3372885.3373835}
\showDOI{\tempurl}


\bibitem[Hu and Carette(2021)]%
        {agda-categories}
\bibfield{author}{\bibinfo{person}{Jason Z.~S. Hu} {and}
  \bibinfo{person}{Jacques Carette}.} \bibinfo{year}{2021}\natexlab{}.
\newblock \showarticletitle{Formalizing Category Theory in Agda}. In
  \bibinfo{booktitle}{\emph{Proceedings of the 10th ACM SIGPLAN International
  Conference on Certified Programs and Proofs}} (Virtual, Denmark)
  \emph{(\bibinfo{series}{CPP 2021})}. \bibinfo{publisher}{Association for
  Computing Machinery}, \bibinfo{address}{New York, NY, USA},
  \bibinfo{pages}{327–342}.
\newblock
\showISBNx{9781450382991}
\urldef\tempurl%
\url{https://doi.org/10.1145/3437992.3439922}
\showDOI{\tempurl}


\bibitem[Jeannin et~al\mbox{.}(2017)]%
        {JeanninKS17}
\bibfield{author}{\bibinfo{person}{Jean{-}Baptiste Jeannin},
  \bibinfo{person}{Dexter Kozen}, {and} \bibinfo{person}{Alexandra Silva}.}
  \bibinfo{year}{2017}\natexlab{}.
\newblock \showarticletitle{Well-founded coalgebras, revisited}.
\newblock \bibinfo{journal}{\emph{Math. Struct. Comput. Sci.}}
  \bibinfo{volume}{27}, \bibinfo{number}{7} (\bibinfo{year}{2017}),
  \bibinfo{pages}{1111--1131}.
\newblock
\urldef\tempurl%
\url{https://doi.org/10.1017/S0960129515000481}
\showDOI{\tempurl}


\bibitem[Lambek(1968)]%
        {lambek}
\bibfield{author}{\bibinfo{person}{Joachim Lambek}.}
  \bibinfo{year}{1968}\natexlab{}.
\newblock \showarticletitle{A Fixpoint Theorem for complete Categories.}
\newblock \bibinfo{journal}{\emph{Mathematische Zeitschrift}}
  \bibinfo{volume}{103} (\bibinfo{year}{1968}), \bibinfo{pages}{151--161}.
\newblock
\urldef\tempurl%
\url{http://eudml.org/doc/170906}
\showURL{%
\tempurl}


\bibitem[Milius(2010)]%
        {Milius10}
\bibfield{author}{\bibinfo{person}{Stefan Milius}.}
  \bibinfo{year}{2010}\natexlab{}.
\newblock \showarticletitle{A Sound and Complete Calculus for Finite Stream
  Circuits}. In \bibinfo{booktitle}{\emph{Proceedings of the 25th Annual {IEEE}
  Symposium on Logic in Computer Science, {LICS} 2010, 11-14 July 2010,
  Edinburgh, United Kingdom}}. \bibinfo{publisher}{{IEEE} Computer Society},
  \bibinfo{pages}{421--430}.
\newblock
\urldef\tempurl%
\url{https://doi.org/10.1109/LICS.2010.11}
\showDOI{\tempurl}


\bibitem[Milius et~al\mbox{.}(2016)]%
        {mpw16}
\bibfield{author}{\bibinfo{person}{Stefan Milius}, \bibinfo{person}{Dirk
  Pattinson}, {and} \bibinfo{person}{Thorsten Wi{\ss}mann}.}
  \bibinfo{year}{2016}\natexlab{}.
\newblock \showarticletitle{A New Foundation for Finitary Corecursion -- The
  Locally Finite Fixpoint and Its Properties}. In
  \bibinfo{booktitle}{\emph{Proc.~19th International Conference on Foundations
  of Software Science and Computation Structures (FoSSaCS 2016)}}
  \emph{(\bibinfo{series}{LNCS}, Vol.~\bibinfo{volume}{9634})},
  \bibfield{editor}{\bibinfo{person}{Bart Jacobs} {and}
  \bibinfo{person}{Christof L{\"{o}}ding}} (Eds.).
  \bibinfo{publisher}{Springer}, \bibinfo{pages}{107--125}.
\newblock
\urldef\tempurl%
\url{https://doi.org/10.1007/978-3-662-49630-5_7}
\showDOI{\tempurl}


\bibitem[Milius et~al\mbox{.}(2020)]%
        {MiliusPW19}
\bibfield{author}{\bibinfo{person}{Stefan Milius}, \bibinfo{person}{Dirk
  Pattinson}, {and} \bibinfo{person}{Thorsten Wi{\ss}mann}.}
  \bibinfo{year}{2020}\natexlab{}.
\newblock \showarticletitle{A new foundation for finitary corecursion and
  iterative algebras}.
\newblock \bibinfo{journal}{\emph{Inf. Comput.}}  \bibinfo{volume}{271}
  (\bibinfo{year}{2020}), \bibinfo{pages}{104456}.
\newblock
\urldef\tempurl%
\url{https://doi.org/10.1016/j.ic.2019.104456}
\showDOI{\tempurl}


\bibitem[Osius(1974)]%
        {osius1974}
\bibfield{author}{\bibinfo{person}{Gerhard Osius}.}
  \bibinfo{year}{1974}\natexlab{}.
\newblock \showarticletitle{Categorical set theory: A characterization of the
  category of sets}.
\newblock \bibinfo{journal}{\emph{Journal of Pure and Applied Algebra}}
  \bibinfo{volume}{4}, \bibinfo{number}{1} (\bibinfo{year}{1974}),
  \bibinfo{pages}{79--119}.
\newblock
\showISSN{0022-4049}
\urldef\tempurl%
\url{https://doi.org/10.1016/0022-4049(74)90032-2}
\showDOI{\tempurl}


\bibitem[Pitts and Steenkamp(2021)]%
        {PittsSteenkamp21}
\bibfield{author}{\bibinfo{person}{Andrew~M. Pitts} {and}
  \bibinfo{person}{S.~C. Steenkamp}.} \bibinfo{year}{2021}\natexlab{}.
\newblock \showarticletitle{Constructing Initial Algebras Using Inflationary
  Iteration}. In \bibinfo{booktitle}{\emph{Proceedings of the Fourth
  International Conference on Applied Category Theory, {ACT} 2021, Cambridge,
  United Kingdom, 12-16th July 2021}}. \bibinfo{pages}{88--102}.
\newblock
\urldef\tempurl%
\url{https://doi.org/10.4204/EPTCS.372.7}
\showDOI{\tempurl}


\bibitem[Sun et~al\mbox{.}(2019)]%
        {transfinitecoq}
\bibfield{author}{\bibinfo{person}{Tianyu Sun}, \bibinfo{person}{Wensheng Yu},
  {and} \bibinfo{person}{Yaoshun Fu}.} \bibinfo{year}{2019}\natexlab{}.
\newblock \showarticletitle{Formalization of Transfinite Induction in Coq}. In
  \bibinfo{booktitle}{\emph{2019 Chinese Automation Congress (CAC)}}.
  \bibinfo{pages}{1001--1005}.
\newblock
\urldef\tempurl%
\url{https://doi.org/10.1109/CAC48633.2019.8997376}
\showDOI{\tempurl}


\bibitem[Taylor(1996)]%
        {Taylor96}
\bibfield{author}{\bibinfo{person}{Paul Taylor}.}
  \bibinfo{year}{1996}\natexlab{}.
\newblock \showarticletitle{Intuitionistic Sets and Ordinals}.
\newblock \bibinfo{journal}{\emph{J. Symb. Log.}} \bibinfo{volume}{61},
  \bibinfo{number}{3} (\bibinfo{year}{1996}), \bibinfo{pages}{705--744}.
\newblock
\urldef\tempurl%
\url{https://doi.org/10.2307/2275781}
\showDOI{\tempurl}


\bibitem[Taylor(1999)]%
        {taylor1999}
\bibfield{author}{\bibinfo{person}{Paul Taylor}.}
  \bibinfo{year}{1999}\natexlab{}.
\newblock \bibinfo{booktitle}{\emph{Practical Foundations of Mathematics}}.
  \bibinfo{series}{Cambridge studies in advanced mathematics},
  Vol.~\bibinfo{volume}{59}.
\newblock \bibinfo{publisher}{Cambridge University Press}.
\newblock
\showISBNx{978-0-521-63107-5}


\bibitem[Taylor(2023)]%
        {Taylor23}
\bibfield{author}{\bibinfo{person}{Paul Taylor}.}
  \bibinfo{year}{2023}\natexlab{}.
\newblock \bibinfo{booktitle}{\emph{Well Founded Coalgebras and Recursion}}.
\newblock
\urldef\tempurl%
\url{https://www.paultaylor.eu/ordinals/welfcr.pdf}
\showURL{%
\tempurl}
\newblock
\shownote{accessed on 2024-01-27}.


\bibitem[Urbat(2017)]%
        {Urbat17}
\bibfield{author}{\bibinfo{person}{Henning Urbat}.}
  \bibinfo{year}{2017}\natexlab{}.
\newblock \showarticletitle{Finite Behaviours and Finitary Corecursion}. In
  \bibinfo{booktitle}{\emph{7th Conference on Algebra and Coalgebra in Computer
  Science, {CALCO} 2017, June 12-16, 2017, Ljubljana, Slovenia}}.
  \bibinfo{pages}{24:1--24:16}.
\newblock
\urldef\tempurl%
\url{https://doi.org/10.4230/LIPIcs.CALCO.2017.24}
\showDOI{\tempurl}


\bibitem[Zermelo(1904)]%
        {Zermelo04}
\bibfield{author}{\bibinfo{person}{Ernst Zermelo}.}
  \bibinfo{year}{1904}\natexlab{}.
\newblock \showarticletitle{{B}eweis, da{\ss} jede {M}enge wohlgeordnet werden
  kann}.
\newblock \bibinfo{journal}{\emph{Math.~Ann.}}  \bibinfo{volume}{59}
  (\bibinfo{year}{1904}), \bibinfo{pages}{514--516}.
\newblock


\end{thebibliography}

\end{document}